\documentclass[letterpaper,11pt]{article}
\usepackage{amssymb}
\usepackage{amsmath}
\usepackage{amsfonts}
\usepackage{harvard}
\usepackage{geometry}
\usepackage{graphicx}

\newtheorem{theorem}{Theorem}

\newtheorem{corollary}[theorem]{Corollary}

\newtheorem{lemma}{Lemma}

\newtheorem{proposition}[theorem]{Proposition}

\numberwithin{equation}{section}
\numberwithin{lemma}{section}
\numberwithin{theorem}{section}
\newenvironment{proof}[1][Proof]{\noindent \textbf{#1.} }{\  \rule{0.5em}{0.5em}}
\begin{document}

\title{{\huge Diffusion Copulas: Identification and Estimation}\\
\bigskip \bigskip }
\author{{\Large Ruijun Bu}\thanks{%
Department of Economics, Management School, University of Liverpool,
Liverpool, UK. Email: ruijunbu@liv.ac.uk.} \and {\Large Kaddour Hadri}%
\thanks{%
Queen's University Management School, Queen's University Belfast, Belfast,
Northern Ireland, UK. Email: k.hadri@qub.ac.uk.} \and {\Large Dennis
Kristensen}\thanks{%
Department of Economics, University College London, London, UK. Email:
d.kristensen@ucl.ac.uk. } \and \bigskip \bigskip \bigskip \bigskip }
\date{April 2020}
\maketitle

\begin{abstract}
We propose a new semiparametric approach for modelling nonlinear univariate
diffusions, where the observed process is a nonparametric transformation of
an underlying parametric diffusion (UPD). This modelling strategy yields a
general class of semiparametric Markov diffusion models with parametric
dynamic copulas and nonparametric marginal distributions. We provide
primitive conditions for the identification of the UPD parameters together
with the unknown transformations from discrete samples. Likelihood-based
estimators of both parametric and nonparametric components are developed and
we analyze the asymptotic properties of these. Kernel-based drift and
diffusion estimators are also proposed and shown to be normally distributed
in large samples. A simulation study investigates the finite sample
performance of our estimators in the context of modelling US short-term
interest rates. We also present a simple application of the proposed method
for modelling the CBOE volatility index data.

\bigskip



\bigskip \bigskip
\end{abstract}

\pagebreak

\section{Introduction}

Most financial time series have fat tails that standard parametric models
are not able to generate. One forceful argument for this in the context of
diffusion models was provided by A\"{\i}t-Sahalia (1996b) who tested a range
of parametric models against a nonparametric alternative and found that most
standard models were inconsistent with observed features in data.

One popular semiparametric approach that allows for more flexibility in
terms of marginal distributions, and so allowing for fat tails, is to use
the so-called copula models, where the copula is parametric and the marginal
distribution is left unspecified (nonparametric). Joe (1997) showed how
bivariate parametric copulas could be used to model discrete-time stationary
Markov chains with flexible, nonparametric marginal distributions. The
resulting class of semiparametric models are relatively easy to estimate;
see, e.g. Chen and Fan (2006). However, most parametric copulas known in the
literature have been derived in a cross-sectional setting where they have
been used to describe the joint dependence between two random variables with
known joint distribution, e.g. a bivariate $t$-distribution. As such,
existing parametric copulas may be difficult to interpret in terms of the
dynamics they imply when used to model Markov processes. This in turn means
that applied researchers may find it difficult to choose an appropriate
copula for a given time series.

One could have hoped that copulas with a clearer dynamic interpretation
could be developed by starting with an underlying parametric Markov model
and then deriving its implied copula. This approach is unfortunately
hindered by the fact that the stationary distributions of general Markov
chains are not available on closed-form and so their implied dynamic copulas
are not available on closed form either. This complicates both the
theoretical analysis (such as establishing identification) and the practical
implementation of such models.

An alternative approach to modelling fat tails using Markov diffusions is to
specify flexible forms\ for the so-called drift and diffusion term. Such
non-linear features tend to generate fat tails in the marginal distribution
of the process. This approach has been widely used to, for example, model
short-term interest rates; see, e.g., A\"{\i}t-Sahalia (1996a,b), Conley et
al. (1997), Stanton (1997), Ahn and Gao (1999) and Bandi (2002). These
models tend to either be heavily parameterized or involve nonparametric
estimators that suffer from low precision in small and moderate samples.

We here propose a novel class of dynamic copulas that resolves the
above-mentioned issues: We show how copulas can easily be generated from
parametric diffusion processes. The copulas have a clear interpretation in
terms of dynamics since they are constructed from an underlying dynamic
continuous-time process. At the same time, a given copula-based diffusion
can exhibit strong non-linearities in its drift and diffusion term even if
the underlying copula\ is derived from, for example, a linear model.
Furthermore, primitive conditions for identification of the parameters are
derived; and this despite the fact that the copulas are implicit. Finally,
the models can easily be implemented in practice using existing numerical
methods for parametric diffusion processes. This in turn implies that
estimators are easy to compute and do not involve any smoothing parameters;
this is in contrast to existing semi- and nonparametric estimators of
diffusion models.

The starting point of our analysis is to show that there is a one-to-one
correspondence between any given semiparametric Markov copula model and a
model where we observe a nonparametric transformation of an underlying
parametric Markov process. We then restrict attention to parametric Markov
diffusion processes which we refer to as underlying parametric diffusions
(UPD's). Copulas generated from a given UPD has a clear interpretation in
terms of dynamic properties. In particular, standard results from the
literature on diffusion models can be employed to establish mixing
properties and existence of moments for a given model; see, e.g. Chen et al.
(2010). Moreover, we are able to derive primitive conditions for the
parameters of the copula to be identified together with the unknown
transformation.

Once identification has been established, estimation of our copula diffusion
models based on a discretely sampled process proceeds as in the
discrete-time case. One can either estimate the model using a one-step or
two-step procedure: In the one-step procedure, the marginal distribution and
the parameters of the UPD are estimated jointly by sieve-maximum likelihood
methods as advocated by Chen, Wu and Yi (2009). In the two-step approach,
the marginal distribution is first estimated by the empirical cdf, which in
turn is plugged into the likelihood function of the model. This is then
maximized with respect to the parameters of the UPD. We provide an
asymptotic theory for both cases by importing results from Chen, Wu and Yi
(2009) and Chen and Fan (2006), respectively. In particular, we provide
primitive conditions for their high-level assumptions to hold in our
diffusion setting. The resulting asymptotic theory shows $\sqrt{n}$%
-asymptotic normality of the parametric components. Given the estimates of
parametric component, one can obtain semiparametric estimates of the drift
and diffusion functions and we also provide an asymptotic theory for these.

Our modelling strategy has parametric ascendants: Bu et al. (2011), Eraker
and Wang (2015) and Forman and S\o rensen (2014) considered parametric
transformations of UPDs for modelling short-term interest rates, variance
risk premia and molecular dynamics, respectively. We here provide a more
flexible class of models relative to theirs since we leave the
transformation unspecified. At the same time, all the attractive properties
of their models remain valid: The transition density of the observed process
is induced by the UPD and so the estimation of copula-based diffusion models
is computationally simple. Moreover, copula diffusion models can furthermore
be easily employed in asset pricing applications since (conditional) moments
are easily computed using the specification of the UPD. Finally, none of
these papers fully addresses the identification issue and so our
identification results are also helpful in their setting.

There are also similarities between our approach and the one pursued in A%
\"{\i}t-Sahalia (1996a) and Kristensen (2010). They developed two classes of
semiparametric diffusion models where either the drift or the diffusion term
is specified parametrically and the remaining term is left unspecified. The
remaining term is then recovered by using the triangular link between the
marginal distribution, the drift and the diffusion terms that exist for
stationary diffusions. In this way, the marginal distribution implicitly
ties down the dynamics of the observed diffusion process. Unfortunately, it
is very difficult to interpret the dynamic properties of the resulting
semiparametric diffusion model. In contrast, in our setting, the UPD alone
ties down the dynamics of the observed diffusion and so these are much
better understood. The estimation of copula diffusions are also less
computationally burdensome compared to the Pseudo Maximum Likelihood
Estimator (PMLE) proposed in Kristensen (2010).

The remainder of this paper is organized as follows. Section \ref{Sec_Model}
outlines our semiparametric modelling strategy. Section \ref%
{Sec_Identification} investigates the identification issue of our model. In
Section \ref{Sec_Inference}, we discuss the estimators of our model while
Section \ref{sec: asymp theory} investigates their asymptotic properties.
Section \ref{Sec_Simulation} presents a simulation study to examine the
finite sample performance of our estimators. In Section \ref{Sec_Application}%
, we consider a simple empirical application. Some concluding remarks are
given in Section \ref{Sec_Conclusion}. All proofs and lemmas are collected
in Appendices.

\section{Copula-Based Diffusion Models\label{Sec_Model}}

\subsection{Framework\label{SubSec_Model}}

Consider a continuous-time process $Y=\left\{ Y_{t}:t\geq 0\right\} $ with
domain $\mathcal{Y}=\left( y_{l},y_{r}\right) $, where $-\infty \leq
y_{l}<y_{r}\leq +\infty $. We assume that $Y$ satisfies%
\begin{equation}
Y_{t}=V\left( X_{t}\right) ,  \label{eq: Y def}
\end{equation}%
where $V:\mathcal{X\mapsto }\mathcal{Y}$ is a smooth monotonic univariate
function and $X=\left\{ X_{t}:t\geq 0\right\} $ solves the following
parametric SDE:%
\begin{equation}
dX_{t}=\mu _{X}\left( X_{t};\theta \right) dt+\sigma _{X}\left( X_{t};\theta
\right) dW_{t}.  \label{UPD_x}
\end{equation}%
Here, $\mu _{X}\left( x;\theta \right) $ and $\sigma _{X}^{2}\left( x;\theta
\right) \ $are scalar functions that are known up to some unknown parameter
vector $\theta \in \Theta $, where $\Theta $ is the parameter space, while $%
W $ is a standard Brownian motion. We call $X$ the underlying parametric
diffusion (UPD) and let $\mathcal{X}=\left( x_{l},x_{r}\right) $, $-\infty
\leq x_{l}<x_{r}\leq +\infty $, denote its domain.

We call $Y$ a \emph{copula-based diffusion} since its dynamics are
determined by the implied (dynamic) copula of the UPD $X$, as we will
explain below. Given a discrete sample of $Y$, $Y_{i\Delta }$, $i=0,1,\ldots
,n$, where $\Delta >0$ denotes the time distance between observations, we
are then interested in drawing inference regarding the parameter $\theta $
and the function $V$. Note here that we only observe $Y$ while $X$ remains
unobserved since we leave $V$ unspecified (unknown to us). For convenience,
we collect the unknown component in the \textit{structure} $\mathcal{S}%
\equiv \left( \theta ,V\right) $.

The above class of models allows for added flexibility through the
transformation $V$ which we treat as a nonparametric object that we wish to
estimate together with $\theta $. By allowing for a broad nonparametric
class of transformations $V$, our model is richer and more flexible compared
to the fully parametric case with known or parametric specifications of $V$.
In particular, as we shall see, any given member of the above class of
models is able to completely match the marginal distribution of any given
time series.

We will require that the underlying Markov process $X$ sampled at $i\Delta $%
, $i=1,2,...$, possesses a transition density $p_{X}\left( x|x_{0};\theta
\right) $,%
\begin{equation}
\Pr \left( X_{\Delta }\in \mathcal{A}|X_{0}=x_{0}\right) =\int_{\mathcal{A}%
}p_{X}\left( x|x_{0};\theta \right) dx,\text{ \ \ }\mathcal{A}\subseteq 
\mathcal{X}.  \label{eq: p_X def}
\end{equation}%
Moreover, some of our results require $X$ to be recurrent, a property which
can be stated in terms of the so-called scale density and scale measure.
These are defined as 
\begin{equation}
s\left( x;\theta \right) :=\exp \left\{ -\int_{x^{\ast }}^{x}\frac{2\mu
_{X}\left( z;\theta \right) }{\sigma _{X}^{2}\left( z;\theta \right) }%
dz\right\} \text{ and }S\left( x;\theta \right) :=\int_{x^{\ast
}}^{x}s\left( z;\theta \right) dz  \label{eq: scale density}
\end{equation}%
for some $x^{\ast }\in \mathcal{X}$. We then impose the following:

\begin{description}
\item[Assumption 2.1.] (i) $\mu _{X}\left( \cdot ;\theta \right) $ and $%
\sigma _{X}^{2}\left( \cdot ;\theta \right) >0$ are twice continuously
differentiable; (ii) the scale measure satisfies $S\left( x;\theta \right)
\rightarrow -\infty \ $($+\infty $) as$\ x\rightarrow x_{l}$ ($x_{r}$);
(iii) $\xi \left( \theta \right) =\int_{\mathcal{X}}\left\{ \sigma
_{X}^{2}\left( x;\theta \right) s\left( x;\theta \right) \right\}
^{-1}dx<\infty $.

\item[Assumption 2.2.] The transformation $V$ is strictly increasing with
inverse $U=V^{-1}$, i.e., $y=V\left( x\right) \Leftrightarrow x=U\left(
y\right) $, and is twice continuously differentiable.
\end{description}

\medskip

Assumption 2.1(i) provides primitive conditions for a solution to eq. (\ref%
{UPD_x}) to exist and for the transition density $p_{X}\left( x|x_{0};\theta
\right) $ to be well-defined, while Assumption 2.1(ii) implies that this
solution is positive recurrent; see Bandi and Phillips (2003), Karatzas and
Shreve (1991, Section 5.5) and McKean (1969, Section 5) for more details.
Assumption 2.1(iii) strengthens the recurrence property to stationarity and
ergodicity in which case the stationary marginal density of $X$\ takes the
form%
\begin{equation}
f_{X}\left( x;\theta \right) =\frac{\xi \left( \theta \right) }{\sigma
_{X}^{2}\left( x;\theta \right) s\left( x;\theta \right) },  \label{mpdf_x}
\end{equation}%
where $\xi \left( \theta \right) $ was defined in Assumption 2.1(iii).
However, stationarity will not be required for all our results to hold; in
particular, some of our identification results and proposed estimators do
not rely on stationarity. This is in contrast to the existing literature on
dynamic copula models where stationarity is a maintained assumption.

Assumption 2.2 requires $V$ to be strictly increasing; this is a testable
restriction under the remaining assumptions introduced below which ensures
identification: Suppose that indeed $V$ is strictly decreasing; we then have 
$Y_{t}=\bar{V}\left( \bar{X}_{t}\right) $, where $\bar{V}\left( x\right)
=V\left( -x\right) $ is increasing and $\bar{X}_{t}=-X_{t}$ has dynamics $%
p_{X}\left( -x|-x_{0};\theta \right) $. Assuming that the chosen UPD
satisfies $p_{X}\left( -x|-x_{0};\theta \right) \neq p_{X}(x|x_{0};\tilde{%
\theta})$ for $\theta \neq \tilde{\theta}$, we can test whether $V$ indeed
is decreasing or increasing.

The smoothness condition on $V$ is imposed so that we can employ Ito's Lemma
on the transformation to obtain that the continuous-time dynamics of $Y$ can
be written in terms of $\mathcal{S}$ as%
\begin{equation*}
dY_{t}=\mu _{Y}\left( Y_{t};\mathcal{S}\right) dt+\sigma _{Y}\left( Y_{t};%
\mathcal{S}\right) dW_{t},
\end{equation*}%
with%
\begin{eqnarray}
\mu _{Y}\left( y;\mathcal{S}\right) &=&\frac{\mu _{X}\left( U\left( y\right)
;\theta \right) }{U^{\prime }\left( y\right) }-\frac{1}{2}\sigma
_{X}^{2}\left( U\left( y\right) ;\theta \right) \frac{U^{\prime \prime
}\left( y\right) }{U^{\prime }\left( y\right) ^{3}},  \label{RDdrift} \\
\sigma _{Y}\left( y;\mathcal{S}\right) &=&\frac{\sigma _{X}\left( U\left(
y\right) ;\theta \right) }{U^{\prime }\left( y\right) },  \label{RDdiffusion}
\end{eqnarray}%
where we have used that, with $U^{\prime }\left( y\right) $ and $U^{\prime
\prime }\left( y\right) $ denoting the first two derivatives of $U\left(
y\right) $, $V^{\prime }\left( U\left( y\right) \right) =1/U^{\prime }\left(
y\right) $ and $V^{\prime \prime }\left( U\left( y\right) \right)
=-U^{\prime \prime }\left( y\right) /U^{\prime }\left( y\right) ^{3}$. In
particular, $Y$ is a Markov diffusion process. As can be seen from the above
expressions, the dynamics of $Y$, as characterized by $\mu _{Y}$ and $\sigma
_{Y}^{2}$, may appear quite complex with $U$ potentially generating
nonlinearities in both the drift and diffusion terms even if $\mu _{X}$ and $%
\sigma _{X}^{2}$ are linear. We demonstrate this feature in the subsequent
subsection where we present examples of simple UPD's are able to generate
non-linear shapes of $\mu _{Y}$\ and $\sigma _{Y}^{2}$\ via the non-linear
transformation $V$.\textbf{\ }At the same time, if we transform $Y$ by $U$
we recover the dynamics of the UPD. As a consequence, the transition density
of the discretely sampled process $Y_{i\Delta }$, $i=0,1,2,...$, can be
expressed in terms of the one of $X$ as%
\begin{equation}
p_{Y}\left( y|y_{0};\mathcal{S}\right) =U^{\prime }\left( y\right)
p_{X}\left( U\left( y\right) |U\left( y_{0}\right) ;\theta \right) ,
\label{tpdf_y(x)}
\end{equation}%
using standard results for densities of invertible transformations. By
similar arguments, the stationary density of $Y$ satisfies 
\begin{equation}
f_{Y}\left( y;\mathcal{S}\right) =U^{\prime }\left( y\right) f_{X}\left(
U\left( y\right) ;\theta \right) ,  \label{mpdf_y(x)}
\end{equation}%
which shows that any choice for UPD is able to fully adapt to any given
marginal density of $Y$ due to the nonparametric nature of $U$.

The above expressions also highlights the following additional theoretical
and practical advantages of our modelling strategy: First, for a given
choice of $U$, we can easily compute $p_{Y}\left( y|y_{0};\mathcal{S}\right) 
$ and $f_{Y}\left( y;\mathcal{S}\right) $ since computation of parametric
transition densities and stationary densities of diffusion models is in
general straightforward, even if they are not available on closed form.
Second, $Y$ inherits all its dynamic properties from $X$; and in the
modelling of $X$, we can rely on a large literature on parametric modelling
of diffusion models. Formally, we have the following straightforward results
adopted from Forman and S\o rensen (2014).

\begin{proposition}
\label{Th: Y prop}Suppose that Assumptions 2.1(i)--(ii) and 2.2 hold. Then
the following results hold for the model (\ref{eq: Y def})-(\ref{UPD_x}):

\begin{enumerate}
\item If Assumption 2.1(iii) hold, then $X$ is stationary and ergodic and so
is $Y$.

\item The mixing coefficients of $X$ and $Y$ coincide.

\item If $E\left[ \left\vert X_{t}\right\vert ^{q_{1}}\right] <\infty $ and $%
\left\vert V\left( x\right) \right\vert \leq B\left( 1+\left\vert
x\right\vert ^{q_{2}}\right) $ for some $B<\infty $ and $q_{1},q_{2}\geq 0$,
then $E[\left\vert Y_{t}\right\vert ^{q_{1}/q_{2}}]<\infty $.

\item If $\varphi $ is an eigenfunction of $X$ with corresponding eigenvalue 
$\rho $ in the sense that $E\left[ \varphi \left( X_{1}\right) |X_{0}\right]
=\rho \varphi \left( X_{0}\right) $ then $\varphi \circ U$ is an
eigenfunction of $Y$ with corresponding eigenvalue $\rho $.
\end{enumerate}
\end{proposition}

The above theorem shows that, given knowledge (or estimates) of $\mathcal{S}$%
, the properties of $Y$ in terms of mixing coefficients, moments, and
eigenfunctions are well-understood since they are inherited from the
specification of $X$. In addition, computations of conditional moments of $Y$
can be done straightforwardly utilizing knowledge of the UPD. For example,
for a given function $G$, the corresponding conditional moment can be
computed as%
\begin{equation*}
E\left[ G\left( Y_{t+s}\right) |Y_{t}=y\right] =E\left[ G_{X}\left(
X_{t+s}\right) |X_{t}=U\left( y\right) \right] ,\text{ where }G_{X}\left(
x\right) :=G\left( V\left( x\right) \right) .
\end{equation*}%
The right-hand side moment only involves $X$ and so standard methods for
computing moments of parametric diffusion models (e.g., Monte Carlo methods,
solving partial differential equations, Fourier transforms) can be employed.
This facilitates the use of our diffusion models in asset pricing where the
price often takes the form of a conditional moment. We refer to Eraker and
Wang (2015) for more details on asset pricing applications for our class of
models; they take a fully parametric approach but all their arguments carry
over to our setting.

The last result of the above theorem will prove useful for our
identification arguments since these will rely on the fundamental
nonparametric identification results derived in Hansen et al. (1998). Their
results involve the spectrum of the observed diffusion process, and the last
result of the theorem implies that the spectrum of $Y$ is fully
characterized by the spectrum of $X$ together with the transformation. The
eigenfunctions and their eigenvalues are also useful for evaluating long-run
properties of $Y$. In our semiparametric approach, the eigenfunctions and
corresponding eigenvalues of $Y$ are easily computed from $X$ and so we
circumvent the problem of estimating these nonparametrically as done in, for
example, Chen, Hansen and Scheinkman (2009) and Gobet et al.\ (2004).

\subsection{Examples of UPDs}

Our framework is quite flexible and in principle allows for any
specification of the UPD for $X$. Many parametric models are available for
that purpose, and we here present three specific examples from the
literature on continuous-time interest rate modelling.\medskip

\noindent \textbf{Example 1: Ornstein-Uhlenbeck (OU) model. }The OU model
(c.f. Vasicek, 1977) is given by%
\begin{equation}
dX_{t}=\kappa \left( \alpha -X_{t}\right) dt+\sigma dW_{t},  \label{OU}
\end{equation}%
defined on the domain $\mathcal{X}=\left( -\infty ,+\infty \right) $. The
process is stationary if and only if $\kappa >0$, in which case $X$
mean-reverts to its unconditional mean $\alpha $. The scale of $X$ is
controlled by $\sigma $. Its stationary and transition distributions are
both normal, and the corresponding copula of the discretely sampled process
is a Gaussian copula with correlation parameter $e^{-\kappa \Delta }$. For
this particular model, the resulting drift and diffusion term of the
observed process takes the form%
\begin{equation}
\mu _{Y}\left( y;\mathcal{S}\right) =\frac{\kappa \left( \alpha -U\left(
y\right) \right) }{U^{\prime }\left( y\right) }-\frac{1}{2}\sigma ^{2}\frac{%
U^{\prime \prime }\left( y\right) }{U^{\prime }\left( y\right) ^{3}},\text{
\ \ }\sigma _{Y}^{2}\left( y;\mathcal{S}\right) =\frac{\sigma ^{2}}{%
U^{\prime }\left( y\right) ^{2}}.
\end{equation}%
In Figure 2 (found in Section \ref{Sec_Simulation}), we plot these two
functions with $U$\ and $\theta $ fitted to the 7-day Eurodollar interest
rate time series\ used in A\"{\i}t-Sahalia (1996b). Observe that $U$\
generates non-linear behavior in $\mu _{Y}$ and $\sigma _{Y}^{2}$ despite
the UPD being a linear Gaussian process.

\medskip

\noindent \textbf{Example 2: Cox-Ingersoll-Ross (CIR) model. }The CIR
process (c.f. Cox et al., 1985) is given by%
\begin{equation}
dX_{t}=\kappa \left( \alpha -X_{t}\right) dt+\sigma \sqrt{X_{t}}dW_{t}.
\label{CIR}
\end{equation}%
The process has domain $\mathcal{X}=\left( 0,+\infty \right) $ and is
stationary if and only if $\kappa >0$, $\alpha >0$ and $2\kappa \alpha
/\sigma ^{2}\geq 1$. Conditional on $X_{i\Delta }$, $X_{\left( i+1\right)
\Delta }$ admits a non-central $\chi ^{2}$ distribution with fractional
degrees of freedom while its stationary distribution is a Gamma
distribution. To our best knowledge, the corresponding dynamic copula has
not been analyzed before or used in empirical work. Figure 4 (in Section \ref%
{Sec_Simulation}) displays $\mu _{Y}$\ and $\sigma _{Y}^{2},$ with $U$\ and $%
\theta $ chosen in the same way as in Exampe 1. Compared to this example,
the resulting drift and diffusion term of $Y$\ exhibit even stronger
non-linearities.

\medskip

\noindent \textbf{Example 3: Nonlinear Drift Constant Elasticity Variance
(NLDCEV) model. }The NLDCEV specification (c.f. Conley et al., 1997) is
given by%
\begin{equation}
dX_{t}=\left( \sum_{i=-k}^{l}\alpha _{i}X_{t}^{i}\right) dt+\sigma
X_{t}^{\beta }dW_{t}  \label{NLDCEV}
\end{equation}%
with domain $\mathcal{X}=\left( 0,+\infty \right) $. It is easily seen that
when $\alpha _{-k}>0$ and $\alpha _{l}<0$ the drift term of the diffusion in
(\ref{NLDCEV}) exhibits mean-reversions for large and small values of $X$. A
popular choice for various studies in finance assumes that $k=1$ and $l=2$
or $3$ (c.f. A\"{\i}t-Sahalia, 1996b; Choi, 2009; Kristensen, 2010; Bu,
Cheng and Hadri, 2017), in which case the drift has linear or zero
mean-reversion in the middle part and much stronger mean-reversion for large
and small values of $X$. Meanwhile, the CEV diffusion term is also
consistent with most empirical findings of the shape of the diffusion term.
It follows that since (\ref{NLDCEV}) is one of the most flexible parametric
diffusions, diffusion processes that are unspecified transformations of (\ref%
{NLDCEV}) should represent a very flexible class of diffusion models.
Similar to (\ref{CIR}), the implied copula of the NLDCEV is new to the
copula literature.\medskip

Examples 1-2 are attractive from a computational standpoint since the
corresponding transition densities are available on closed-form thereby
facilitating their implementation. But this comes at the cost of the
dynamics being somewhat simple. The NLDCEV\ model implies more complex and
richer dynamics but on the other hand its transition density is not
available on closed form. However, the marginal pdf of the NLDCEV process,
as well as more general specifications,\ can be evaluated in closed form by (%
\ref{mpdf_x}). Moreover, closed-from approximations of the transition
density of the NLDCEV model developed by, for example, A\"{\i}t-Sahalia
(2002) and Li (2013) can be employed. Alternatively, simulated versions of
the transition density can be computed using the techniques developed in,
for example, Kristensen and Shin (2012) and Bladt and S\o rensen (2014). In
either case, an approximate version of the exact likelihood can be easily
computed, thereby allowing for simple estimation of even quite complex
underlying UPDs.

\subsection{Related Literature}

As already noted in the introduction, copula-based diffusions are related to
the class of so-called \emph{discrete-time} copula-based Markov models; see,
for example, Chen and Fan (2006) and references therein. To map the notation
and ideas of this literature into our continuous-time setting, we set the
sampling time distance $\Delta =1$ in the remaining part of this section.

Let us first introduce copula-based Markov models where a given
discrete-time, stationary scalar Markov process $Y=\left\{
Y_{i}:i=0,1,\ldots ,n\right\} $ is modelled through a bivariate parametric
copula density\footnote{%
The copula $C_{X}\left( u_{0},u_{1};\theta \right) $ for a given Markov
process is defined as%
\begin{equation*}
C_{X}\left( u_{0},u_{1};\theta \right) =\Pr \left( X_{0}\leq
F_{X}^{-1}\left( u_{0};\theta \right) ,X_{1}\leq F_{X}^{-1}\left(
u_{1};\theta \right) \right) .
\end{equation*}%
The corresponding copula density is then given by $c_{X}\left(
u_{0},u_{1};\theta \right) =\partial ^{2}C_{X}\left( u_{0},u_{1};\theta
\right) /\left( \partial u_{0}\partial u_{1}\right) $.}, say, $c_{X}\left(
u_{0},u;\theta \right) $, together with its stationary marginal cdf $F_{Y}$,
i.e., so that $Y$'s transition density satisfies%
\begin{equation}
p_{Y}\left( y|y_{0};\theta ,F_{Y}\right) =f_{Y}\left( y\right) c_{X}\left(
F_{Y}\left( y_{0}\right) ,F_{Y}\left( y\right) ;\theta \right) ,
\label{eq: copula MC}
\end{equation}%
where $f_{Y}\left( y\right) =F_{Y}^{\prime }\left( y\right) $. An
alternative representation of this model is%
\begin{equation}
Y_{i}=F_{Y}^{-1}\left( \bar{X}_{i}\right) ,\text{ \ \ }\bar{X}_{i+1}|\bar{X}%
_{i}=x_{0}\sim c_{X}\left( x_{0},\cdot ;\theta \right) ,
\label{eq: copula MC 2}
\end{equation}%
so that $Y_{i}$ is a transformation of an underlying Markov process $\bar{X}%
_{i}\in \left[ 0,1\right] $; the latter having a uniform marginal
distribution and transition density $c_{X}\left( x_{0},x;\theta \right) $.
Thus, if $c_{X}\left( x_{0},x;\theta \right) $ is induced by an underlying
Markov diffusion transition density, the corresponding copula-based Markov
model falls within our framework.

Reversely, consider a copula-based diffusion and suppose that the UPD $X$ is
stationary with marginal cdf $F_{X}\left( x;\theta \right) $. By definition
of $Y$, its marginal cdf satisfies 
\begin{equation}
F_{Y}\left( y\right) =F_{X}\left( U\left( y\right) ;\theta \right)
\Leftrightarrow U\left( y\right) =F_{X}^{-1}\left( F_{Y}\left( y\right)
;\theta \right) .  \label{U}
\end{equation}%
Substituting the last expression for $U$ into (\ref{tpdf_y(x)}), we see that 
$p_{Y}$ can be expressed in the form of (\ref{eq: copula MC}) where $%
c_{X}\left( u_{0},u;\theta \right) $ is the density function of the
(dynamic) copula implied by the discretely sampled UPD $X$,%
\begin{equation}
c_{X}\left( u_{0},u;\theta \right) =\frac{p_{X}\left( F_{X}^{-1}\left(
u;\theta \right) |F_{X}^{-1}\left( u_{0};\theta \right) ;\theta \right) }{%
f_{X}\left( F_{X}^{-1}\left( u;\theta \right) ;\theta \right) }.  \label{dic}
\end{equation}%
Thus, any discretely sampled stationary copula-based diffusion satisfies (%
\ref{eq: copula MC 2}) with $\bar{X}_{i}=F_{X}\left( X_{i}\right) $.

However, the literature on copula-based Markov models focus on discrete-time
models with standard copula specifications derived from bivariate
distributions in an i.i.d. setting. Using copulas that are originally
derived in an i.i.d. setting complicates the interpretation of the dynamics
of the resulting Markov model, and conditions for the model to be mixing,
for example, can be quite complicated to derive; see, e.g., Beare (2010) and
Chen, Wu and Yi (2009). This also implies that very few standard copulas can
be interpreted as diffusion processes; to our knowledge, the only one is the
Gaussian copula which corresponds to the OU process in Example 1.

The reader may now wonder why we do not simply generate dynamic copulas by
first deriving the transition density $p_{X}\left( x|x_{0};\theta \right) $
for a given discrete-time Markov model and then obtain the corresponding
Markov copula through eq. (\ref{dic})? The reason is that for most
discrete-time Markov models the stationary distribution $F_{X}\left(
x;\theta \right) $ is not known on closed form. Thus, first of all, $%
F_{X}^{-1}\left( u;\theta \right) $ and thereby also $c_{X}$ have be
approximated numerically. Second, since $c_{X}$ is now not available on
closed form, the analysis of which parameters one can identify from the
resulting copula model becomes very challenging. And identification in
copula-based Markov models is a non-trivial problem: Generally, for a given
parametric Markov model, not all parameters are identified from the
corresponding copula as given in (\ref{dic}) and some of them have to be
normalized.

We here directly generate copulas through an underlying continuous-time
diffusion model for $X$. This resolves the aforementioned drawbacks of
existing copula-based Markov models: First, we are able to generate highly
flexible copulas so far not considered in the literature. Second, given that
our copulas are induced by specifying the drift and diffusion functions of $%
X $, the time series properties are much more easily inferred from our
model, c.f. Theorem \ref{Th: Y prop}. Third, by Ito's Lemma, eqs. (\ref%
{RDdrift})-(\ref{RDdiffusion}) provide us with explicit expressions linking
the drift and diffusion terms of the observed diffusion process $Y$ to the
UPD through the transformation $V$; this will allow us to derive necessary
and sufficient conditions for identification in the following. Fourth, in
terms of estimation, the stationary distribution of a given diffusion model
has an explicit form, c.f. eq. (\ref{mpdf_x}), which allows us to develop
computationally simple estimators of copula diffusion models. Finally, some
of our identification results will not require stationarity and so expands
the scope for using copula-type models in time series analysis.

Our modelling strategy is also related to the ideas of A\"{\i}t-Sahalia
(1996a) and Kristensen (2010, 2011) where $F_{Y}$ is left unspecified while
either the drift, $\mu _{Y}$, or the diffusion term, $\sigma _{Y}^{2}$, is
specified parametrically. As an example, consider the former case where $%
\sigma _{Y}^{2}\left( y;\theta \right) $ is known up to the parameter $%
\theta $. Given knowledge of the marginal density $f_{Y}$ (or a
nonparametric estimator of it), the diffusion term can then be recovered as
a functional of $f_{Y}$ and $\mu _{Y}$ as%
\begin{equation*}
\mu _{Y}\left( y;f_{Y},\theta \right) =\frac{1}{2f_{Y}\left( y\right) }\frac{%
\partial }{\partial y}\left[ \sigma _{Y}^{2}\left( y;\theta \right)
f_{Y}\left( y\right) \right] .
\end{equation*}%
So in their setting $f_{Y}$ pins down the resulting dynamics of $Y$ in a
rather opaque manner.

\section{Identification\label{Sec_Identification}}

Suppose that a particular specification of the UPD as given in (\ref{UPD_x})
has been chosen. Given the discrete sample of $Y$, the goal is to obtain
consistent estimates of $\theta $ together with $V$. To this end, we first
have to show that these are actually identified from data. In order to do
so, we need to be precise about which primitives we can identify from data.
Given the primitives, we then wish to recover $\left( \theta ,V\right) $. In
the cross-sectional literature, one normally take as given the distribution
of data and then establish a mapping between this and the structural
parameters. In our setting, we are able to learn about the transition
density of our data, $p_{Y}$, from the population and so it would be natural
to use this as primitive from which we wish to recover $\left( \theta
,V\right) $. However, the mapping from $p_{Y}$ to $\left( \theta ,V\right) $
is not available on closed form in general in our setting and so this
identification strategy appears highly complicated. Instead we will take as
primitives the drift, $\mu _{Y}$, and diffusion term, $\sigma _{Y}^{2}$, of $%
Y$ and then show identification of $\left( \theta ,V\right) $ from these.
This identification argument relies on us being able to identify $\mu _{Y}$
and $\sigma _{Y}^{2}$ in the first place, which we formally assume here:

\begin{description}
\item[Assumption 3.1] The drift, $\mu _{Y}$, and the diffusion, $\sigma
_{Y}^{2}$, are nonparametrically identified from the discretely sampled
process $Y$.
\end{description}

The above assumption is not completely innocuous and does impose some
additional regularity conditions on the Data Generating Process (DGP). We
therefore first provide sufficient conditions under which Assumption 3.1
holds. The first set of conditions are due to Hansen et al. (1998) who
showed that Assumption 3.1 is satisfied if $Y$ is stationary and its
infinitesimal operator has a discrete spectrum. Theorem \ref{Th: Y prop}(4)
is helpful in this regard since it informs us that the spectrum of $Y$ can
be recovered from the one of $X$. In particular, if $X$ is stationary with a
discrete spectrum, then $Y$ will have the same properties. Since the
dynamics of $X$ is known to us, the properties of its spectrum are in
principle known to us and so this condition can be verified a priori. The
second set of primitive conditions come from Bandi and Phillips (2003): They
show that as $\Delta \rightarrow 0$ and $n\Delta \rightarrow \infty $, the
drift and diffusion functions of a recurrent Markov diffusion process are
identified. This last result holds without stationarity, but on the other
hand requires high-frequency observations.

In order to formally state the above two results, we need some additional
notation. Recall that the infinitesimal operator, denoted $L_{X}$, of a
given UPD $X$ is defined as%
\begin{equation*}
L_{X,\theta }g\left( x\right) :=\mu _{X}\left( x;\theta \right) g^{\prime
}\left( x\right) +\frac{1}{2}\sigma _{X}^{2}\left( x;\theta \right)
g^{\prime \prime }\left( x\right) ,
\end{equation*}%
for any twice differentiable function $g\left( x\right) $. We follow Hansen
et al. (1998) and restrict the domain of $L_{X}$ to the following set of
functions:%
\begin{equation*}
\mathcal{D}\left( L_{X,\theta }\right) =\left\{ g\in L_{2}\left(
f_{X}\right) :g^{\prime }\text{ is a.c., }L_{X,\theta }g\in L_{2}\left(
f_{X}\right) \text{ and }\lim_{x\downarrow x_{l}}\frac{g^{\prime }\left(
x\right) }{s\left( x\right) }=\lim_{x\uparrow x_{u}}\frac{g^{\prime }\left(
x\right) }{s\left( x\right) }=0\right\} .
\end{equation*}%
where a.c. stands for absolutely continuous. The spectrum of $L_{X,\theta }$
is then the set of solution pairs $\left( \varphi ,\rho \right) $, with $%
\varphi \in \mathcal{D}\left( L_{X,\theta }\right) $ and $\rho \geq 0$, to
the following eigenvalue problem, $L_{X,\theta }\varphi =-\rho \varphi $. We
refer to Hansen et al. (1998) and Kessler and S\o rensen (1999) for a
further discussion\ and results regarding the spectrum of $L_{X}$. The
following result then holds:

\begin{proposition}
Suppose that Assumption 2.1(i)-(ii) is satisfied. Then Assumption 3.1 holds
under either of the following two sets of conditions:

\begin{enumerate}
\item Assumption 2.1(iii) holds and $L_{X,\theta }$ has a discrete spectrum
where $\theta $ is the data-generating parameter value.

\item $\Delta \rightarrow 0$ and $n\Delta \rightarrow \infty $.
\end{enumerate}
\end{proposition}

Importantly, the above result shows that Assumption 3.1 can be verified
without imposing stationarity. Unfortunately, this requires high-frequency
information ($\Delta \rightarrow 0$). To our knowledge, there exists no
results for low-frequency ($\Delta >0$ fixed)\ identification of the drift
and diffusion terms of scalar diffusion processes under non-stationarity.
But by inspection of the arguments of Hansen et al. (1998) one can verify
that at least the diffusion component is nonparametrically identified from
low-frequency information without stationarity.

We are now ready to analyze the identification problem. Recall that $%
\mathcal{S}=\left( \theta ,V\right) $ contains the objects of interest and
let our model consist of all the structures that satisfy, as a minimum,
Assumptions 2.1(i)--(ii) and 2.2. According to (\ref{RDdrift})-(\ref%
{RDdiffusion}), each structure implies a drift and diffusion term of the
observed process. We shall say that two structures $\mathcal{S}=\left(
\theta ,V\right) $ and $\mathcal{\tilde{S}}=(\tilde{\theta},\tilde{V})$ are 
\emph{observationally equivalent}, a property which we denote by $\mathcal{S}%
\sim \mathcal{\tilde{S}}$, if they imply the same drift and diffusion of $Y$%
, i.e. 
\begin{equation}
\forall y\in \mathcal{Y}:\mu _{Y}\left( y;\mathcal{S}\right) =\mu _{Y}\left(
y;\mathcal{\tilde{S}}\right) \text{ and }\sigma _{Y}\left( y;\mathcal{S}%
\right) =\sigma _{Y}\left( y;\mathcal{\tilde{S}}\right) .
\end{equation}%
The structure $\mathcal{S}$ is then said to be identified within the model
if $\mathcal{S}\sim \mathcal{\tilde{S}}$ implies $\mathcal{S}=\mathcal{%
\tilde{S}}$. In our setting, without suitable \emph{normalizations} on the
parameters of the UPD, identification will generally fail. To see this,
observe that any given structure $\mathcal{S}$ is observationally equivalent
to the following process: Choose any one-to-one transformation $T:\mathcal{X}%
\mapsto \mathcal{X}$, and rewrite the DGP implied by $\mathcal{S}$ as%
\begin{equation}
Y_{t}=\tilde{V}\left( \tilde{X}_{t}\right) ,\text{ \ \ }\tilde{V}\left(
x\right) =V\left( T\left( x\right) \right) ,  \label{eq: Y = V(tilde X)}
\end{equation}%
where $\tilde{X}_{t}=T^{-1}\left( X_{t}\right) $ solves%
\begin{equation}
d\tilde{X}_{t}=\mu _{T^{-1}\left( X\right) }\left( \tilde{X}_{t};\theta
\right) dt+\sigma _{T^{-1}\left( X\right) }\left( \tilde{X}_{t};\theta
\right) dW_{t},  \label{eq: tilde X def}
\end{equation}%
with%
\begin{eqnarray}
\mu _{T^{-1}\left( X\right) }\left( x;\theta \right) &=&\frac{\mu _{X}\left(
T\left( x\right) ;\theta \right) }{\partial T\left( x\right) /\left(
\partial x\right) }-\frac{1}{2}\sigma _{X}^{2}\left( T\left( x\right)
;\theta \right) \frac{\partial ^{2}T\left( x\right) /\left( \partial
x^{2}\right) }{\partial T\left( x\right) /\left( \partial x\right) ^{3}},
\label{eq: mu tilde X} \\
\sigma _{T^{-1}\left( X\right) }\left( x;\theta \right) &=&\frac{\sigma
_{X}\left( T\left( x\right) ;\theta \right) }{\partial T\left( x\right)
/\left( \partial x\right) }.  \label{eq: sig tilde X}
\end{eqnarray}%
Suppose now that there exists $\tilde{\theta}$ so that $\mu _{T^{-1}\left(
X\right) }\left( x;\theta \right) =\mu _{X}\left( x;\tilde{\theta}\right) $
and $\sigma _{T^{-1}\left( X\right) }\left( x;\theta \right) =\sigma
_{X}\left( x;\tilde{\theta}\right) $. Then the alternative representation (%
\ref{eq: Y = V(tilde X)})-(\ref{eq: tilde X def}) is a member of our model
with structure $\mathcal{\tilde{S}=(}\tilde{\theta},\tilde{V})$ which is
observationally equivalent to $\mathcal{S}=\left( V,\theta \right) $. The
following result provides a complete characterizations of the class of
observationally equivalent structures for a given model:

\begin{theorem}
\label{Th: Master ID}Suppose that Assumptions 3.1 is satisfied. For any two
structures $\mathcal{S}=\left( V,\theta \right) $ and $\mathcal{\tilde{S}}=(%
\tilde{V},\tilde{\theta})$ satisfying Assumptions 2.1(i) and 2.2, the
following hold: $\mathcal{S}\sim \mathcal{\tilde{S}}$ if and only if there
exists one-to-one transformation $T:\mathcal{X}\mapsto \mathcal{X}$ so that%
\begin{equation}
\tilde{V}\left( x\right) =V\left( T\left( x\right) \right)
\label{eq: equiv 1}
\end{equation}%
and, with $\mu _{T^{-1}\left( X\right) }\left( x;\tilde{\theta}\right) $ and 
$\sigma _{T^{-1}\left( X\right) }\left( x;\tilde{\theta}\right) $ given in
eqs. (\ref{eq: mu tilde X})-(\ref{eq: sig tilde X}), 
\begin{equation}
\text{(i) }\mu _{T^{-1}\left( X\right) }(x;\tilde{\theta})=\mu _{X}\left(
x;\theta \right) \text{ and (ii) }\sigma _{T^{-1}\left( X\right) }(x;\tilde{%
\theta})=\sigma _{X}\left( x;\theta \right) .  \label{eq: equiv 2}
\end{equation}

In particular, the data-generating structure is identified if and only if
there exists no one-to-one transformation $T$ such that (\ref{eq: equiv 2})
holds for $\theta \neq \tilde{\theta}$.
\end{theorem}

Note that the above theorem does not require stationarity since it is only
concerned with the mapping $\mathcal{S\mapsto }\left( \mu _{Y}\left( \cdot ;%
\mathcal{S}\right) ,\sigma _{Y}\left( \cdot ;\mathcal{S}\right) \right) $
which is well-defined irrespectively of whether data is stationary. The
first part of the theorem provides a exact characterization of when any two
structures are equivalent, namely if there exists a transformation $T$ so
that (\ref{eq: equiv 1})-(\ref{eq: equiv 2}) hold. The second part comes as
a natural consequence of the first part: If there exists no such
transformation, then the data-generating structure must be identified.

Unfortunately, the above result may not always be useful in practice since
it requires us to search over all possible one-to-one transformations $T$
and for each of these verify that there exists no $\theta \neq \tilde{\theta}
$ for which eq. (\ref{eq: equiv 2}) holds. In some cases, it proves useful
to first normalize the UPD suitably and then verify eq. (\ref{eq: equiv 2})
in the normalized version. First note that for any one-to-one transformation 
$\bar{T}\left( \cdot ;\theta \right) :\mathcal{X}\mapsto \mathcal{\bar{X}}$,
an equivalent representation of the model is%
\begin{equation*}
Y_{t}=V\left( \bar{X}_{t}\right) ,
\end{equation*}%
where the "normalised" UPD $\bar{X}_{t}:=\bar{T}^{-1}\left( X_{t};\theta
\right) \in \mathcal{\bar{X}}$ solves%
\begin{equation*}
d\bar{X}_{t}=\mu _{\bar{X}}\left( \bar{X}_{t};\theta \right) dt+\sigma _{%
\bar{X}}\left( \bar{X}_{t};\theta \right) dW_{t},
\end{equation*}%
with%
\begin{eqnarray}
\mu _{\bar{X}}\left( \bar{x};\theta \right) &=&\frac{\mu _{X}\left( \bar{T}%
\left( \bar{x};\theta \right) ;\theta \right) }{\partial \bar{T}\left( \bar{x%
};\theta \right) /\left( \partial \bar{x}\right) }-\frac{1}{2}\sigma
_{X}^{2}\left( \bar{T}\left( \bar{x};\theta \right) ;\theta \right) \frac{%
\partial ^{2}\bar{T}\left( \bar{x};\theta \right) /\left( \partial \bar{x}%
^{2}\right) }{\partial \bar{T}\left( \bar{x};\theta \right) /\left( \partial 
\bar{x}\right) ^{3}},  \label{eq: norm X drift} \\
\sigma _{\bar{X}}\left( \bar{x};\theta \right) &=&\frac{\sigma _{X}\left( 
\bar{T}\left( \bar{x};\theta \right) ;\theta \right) }{\partial \bar{T}%
\left( \bar{x};\theta \right) /\left( \partial \bar{x}\right) }.
\label{eq: norm X diff}
\end{eqnarray}%
Given that the above representation is observationally equivalent to the
original model, we can still employ Theorem \ref{Th: Master ID} but with $%
\mu _{\bar{X}}$ and $\sigma _{\bar{X}}$ replacing $\mu _{X}$ and $\sigma
_{X} $. Verifying the identification conditions stated in the second part of
the theorem for the normalised versions will in some situations be easier by
judicious choice of $\bar{T}$.

Below, we present three particular normalising transformations that we have
found useful in this regard. The chosen transformations allow us to provide
easy-to-check conditions for a given UPD to be identified. For a given UPD,
the researcher is free to apply either of the three identification schemes
depending on which is the easier one to implement. The three schemes lead to
different normalizations/parametrizations, but they all lead to models that
are exactly identified (no over-identifying restrictions are imposed) and so
are observationally equivalent: The resulting form of $\mu _{Y}$\ and $%
\sigma _{Y}$\ will be identical irrespectively of which scheme is employed.

The three transformations that we consider also highlights three alternative
modelling approaches: Instead of starting with a parametric UPD as found in
the existing literature, such as Examples 1-3, one can alternatively build a
UPD with unit diffusion ($\sigma _{X}=1$), zero drift ($\mu _{X}=0$) or
known marginal distribution. As we shall see, either of these three
modelling approaches are in principle as flexible as the standard approach
where the researcher jointly specifies the drift and diffusion term.

\subsection{First Scheme}

In our first identification scheme, we choose to normalize $X_{t}$ by the
so-called Lamperti transform,%
\begin{equation*}
\bar{X}_{t}=\bar{T}^{-1}\left( X_{t};\theta \right) :=\gamma \left(
X_{t};\theta \right) ,\text{ \ \ }\gamma \left( x;\theta \right)
=\int_{x^{\ast }}^{x}\frac{1}{\sigma _{X}\left( z;\theta \right) }dz,
\end{equation*}%
for some $x^{\ast }\in \mathcal{X}$. The resulting process is a unit
diffusion process,%
\begin{equation*}
d\bar{X}_{t}=\mu _{\bar{X}}\left( \bar{X}_{t};\theta \right) dt+dW_{t},
\end{equation*}%
with domain $\mathcal{\bar{X}}=\left( \bar{x}_{l},\bar{x}_{r}\right) $,
where $\bar{x}_{r}=\lim_{x\rightarrow x_{r}^{+}}\gamma \left( x;\theta
\right) $ and $\bar{x}_{l}=\lim_{x\rightarrow x_{l}^{-}}\gamma \left(
x;\theta \right) $, and drift function%
\begin{equation}
\mu _{\bar{X}}\left( \bar{x};\theta \right) =\frac{\mu _{X}\left( \gamma
^{-1}\left( \bar{x};\theta \right) ;\theta \right) }{\sigma _{X}\left(
\gamma ^{-1}\left( \bar{x};\theta \right) ;\theta \right) }-\frac{1}{2}\frac{%
\partial \sigma _{X}}{\partial x}\left( \gamma ^{-1}\left( \bar{x};\theta
\right) ;\theta \right) .  \label{eq: Xbar drift}
\end{equation}%
For the unit diffusion version of the UPD, the equivalence condition (\ref%
{eq: equiv 2})(ii) becomes%
\begin{equation*}
1=\sigma _{\bar{X}}\left( \bar{x};\theta \right) =\sigma _{T^{-1}\left( \bar{%
X}\right) }\left( \bar{x};\tilde{\theta}\right) =\frac{1}{\partial T\left( 
\bar{x}\right) /\left( \partial x\right) },
\end{equation*}%
which can only hold if $T\left( \bar{x}\right) =\bar{x}+\eta $ for some
constant $\eta \in \mathbb{R}$. Thus, we can restrict attention to this
class of transformations and (\ref{eq: equiv 2})(i) becomes:

\begin{description}
\item[Assumption 3.2.] With $\mu _{\bar{X}}$ given in (\ref{eq: Xbar drift}%
): There exists no $\eta \neq 0$ and $\tilde{\theta}\neq \theta $ such that $%
\mu _{\bar{X}}(\bar{x};\tilde{\theta})=\mu _{\bar{X}}\left( \bar{x}+\eta
;\theta \right) $ for all $\bar{x}\in \mathcal{\bar{X}}$.
\end{description}

\bigskip

Assumption 3.2 imposes a normalization condition on the transformed drift
function to ensure identification. When verifying Assumption 3.2 for the
transformed unit diffusion $\bar{X}$ defined above, we will generally need
to fix some of the parameters that enter $\mu _{X}\left( x;\theta \right) $
and $\sigma _{X}^{2}\left( x;\theta \right) $ of the original process $X$,
see below.

\begin{corollary}
\label{Cor: scheme 1}Under Assumptions 2.1(i), 2.2 and 3.1, $\mathcal{S}$ is
identified if and only if Assumption 3.2 is satisfied.
\end{corollary}

\bigskip

The above transformation result can be applied to standard parametric
specifications when $\gamma \left( x;\theta \right) $ is available on
closed-form. But it also highlights that in terms of modelling copula
diffusions, we can without loss of generality build a model where we from
the outset restrict $\sigma _{X}=1$ and only model the drift term $\mu _{X}$%
. For example, we could choose the following flexible polynomial drift model
where we have already normalized the diffusion term:%
\begin{equation}
dX_{t}=\left( \sum_{i=1}^{l}\alpha _{i}X_{t}^{i}\right) dt+dW_{t},
\end{equation}%
where $\theta =\left( \alpha _{1},...,\alpha _{l}\right) $. Corollary \ref%
{Cor: scheme 1} shows that this particular copula diffusion specification is
identified without further restrictions on $\theta $. Below we apply
Corollary \ref{Cor: scheme 1} to some of the standard parametric diffusions
introduced earlier:\bigskip

\noindent \textbf{Example 1 (continued). }The Lamperti transform of the OU\
process in (\ref{OU}) is given by 
\begin{equation*}
d\bar{X}_{t}=\kappa \left( \alpha /\sigma -\bar{X}_{t}\right) dt+dW_{t}.
\end{equation*}%
Since $\alpha /\sigma $ is a location shift of $\bar{X}$, we need to
normalize $\alpha /\sigma $ in order for the identification condition 3.3 to
be satisfied; one such is $\alpha /\sigma =0$ leading to the following
identified model,%
\begin{equation}
d\bar{X}_{t}=-\kappa \bar{X}_{t}dt+dW_{t}.  \label{OU_LT}
\end{equation}

\bigskip

\noindent \textbf{Example 2 (continued). }The Lamperti transform of the CIR\
diffusion in (\ref{CIR}) is given by%
\begin{equation}
d\bar{X}_{t}=\left[ \kappa \left( \frac{2}{\bar{X}_{t}}\frac{\alpha }{\sigma
^{2}}-\frac{\bar{X}_{t}}{2}\right) -\frac{1}{2\bar{X}_{t}}\right] dt+dW_{t},
\label{CIR_LT}
\end{equation}%
which only depends on $\theta =\left( \kappa ,\alpha ^{\ast }\right) $ where 
$\alpha ^{\ast }=\alpha /\sigma ^{2}$. Note that the dimension of the
parameter vector reduced from $3$ to $2$. Crucially, it also suggests that
we can only identify $\alpha $ and $\sigma ^{2}$ up to a ratio. Hence,
normalization requires fixing either $\alpha $, $\sigma ^{2}$, or their
ratio.

\bigskip

\noindent \textbf{Example 3 (continued). }It can be easily verified that the
Lamperti transform of the NLDCEV diffusion in (\ref{NLDCEV}) takes the form%
\begin{equation}
d\bar{X}_{t}=\left[ \sum_{i=-k}^{l}\alpha _{i}^{\ast }\bar{X}_{t}^{\frac{%
i-\beta }{1-\beta }}-\frac{\beta }{2\left( 1-\beta \right) }\bar{X}_{t}^{-1}%
\right] dt+dW_{t},  \label{NLDCEV_LT}
\end{equation}%
where $\alpha _{i}^{\ast }:=\alpha _{i}\sigma ^{\frac{i-1}{1-\beta }}\left(
1-\beta \right) ^{\frac{i-\beta }{1-\beta }}$, $i=-k,...,l$. Hence, the
parameters $\theta =\left( \beta ,\alpha _{-k}^{\ast },...,\alpha
_{-l}^{\ast }\right) $ are identified and the number of parameters is
reduced from $l+k+3$ to $l+k+2$. Note that just as (\ref{OU}) and (\ref{CIR}%
) are special cases of (\ref{NLDCEV}), both (\ref{OU_LT}) and (\ref{CIR_LT})
are special cases of (\ref{NLDCEV_LT}).

\subsection{Second Scheme}

Our second identification strategy transforms $X$ by its scale measure
defined in eq. (\ref{eq: scale density}),%
\begin{equation*}
\bar{X}_{t}:=S\left( X_{t};\theta \right) ,
\end{equation*}
which brings the diffusion process onto its natural scale,%
\begin{equation*}
d\bar{X}_{t}=\sigma _{\bar{X}}\left( \bar{X}_{t};\theta \right) dW_{t},
\end{equation*}%
where the drift is zero (and so known) while%
\begin{equation}
\sigma _{\bar{X}}^{2}\left( \bar{x};\theta \right) =s^{2}\left( S^{-1}\left( 
\bar{x};\theta \right) ;\theta \right) \sigma ^{2}\left( S^{-1}\left( \bar{x}%
;\theta \right) ;\theta \right) .  \label{eq: scale diffusion term}
\end{equation}

Since the drift term is zero, the identification condition (\ref{eq: equiv 2}%
)(i) becomes%
\begin{equation}
0=-\frac{1}{2}\sigma _{\bar{X}}^{2}\left( T\left( \bar{x}\right) ;\tilde{%
\theta}\right) \frac{\partial ^{2}T\left( \bar{x}\right) /\left( \partial 
\bar{x}^{2}\right) }{\partial T\left( \bar{x}\right) /\left( \partial \bar{x}%
\right) ^{3}},
\end{equation}%
which can only hold if $\partial ^{2}T\left( \bar{x}\right) /\left( \partial 
\bar{x}^{2}\right) =0$. We can therefore restrict attention to linear
transformations $T\left( \bar{x}\right) =\eta _{1}\bar{x}+\eta _{2}$, for
some constants $\eta _{1},\eta _{2}\in \mathbb{R}$, in which case (\ref{eq:
equiv 2})(ii) becomes:\bigskip

\begin{description}
\item[Assumption 3.3.] With $\sigma _{\bar{X}}^{2}$ given in (\ref{eq: scale
diffusion term}): There exists no $\eta _{1}\neq 1$, $\eta _{2}\neq 0$ and $%
\tilde{\theta}\neq \theta $ such that $\sigma _{\bar{X}}^{2}(\bar{x};\tilde{%
\theta})=\sigma _{\bar{X}}^{2}\left( \eta _{1}\bar{x}+\eta _{2};\theta
\right) /\eta _{1}^{2}$ for all $\bar{x}\in \mathcal{\bar{X}}$.
\end{description}

\bigskip

In comparison to Assumption 3.2, we here have to impose two normalizations
to ensure identification. The intuition for this is that setting the drift
to zero does not act as a complete normalization of the process:\ Any
additional scale transformation of $\bar{X}$ still leads to a zero-drift
process. Therefore, for the third scheme to work we need both a scale and
location normalization.\bigskip

\begin{theorem}
\label{Cor: scheme 2}Under Assumptions 2.1(i)--(ii), 2.2 and 3.1, $\mathcal{S%
}$ is identified if and only if Assumption 3.3 is satisfied.
\end{theorem}

\bigskip

Compared to the first identification scheme, it\ is noticeably harder to
apply this one to existing parametric diffusion models since the inverse of
the scale transform is usually not available in closed form. But, similar to
the first identification scheme, the result shows that without loss of
flexibility, we can focus on UPDs with zero drift and then model the
diffusion term in a flexible manner, e.g.,%
\begin{equation}
dX_{t}=\exp \left( \sum_{i=1}^{l-1}\beta _{i}X_{t}^{i}+\beta _{l}\left\vert
X_{t}\right\vert ^{l}\right) dW_{t}.
\end{equation}%
Corollary \ref{Cor: scheme 2} shows that this UPD is identified together
with $V$ without any further parameter restrictions on $\theta =\left( \beta
_{1},...,\beta _{l}\right) $.

\subsection{Third scheme}

Our third identification strategy transforms a given stationary UPD by its
marginal cdf,%
\begin{equation}
\bar{X}_{t}=F_{X}\left( X_{t};\theta \right) .  \label{eq: cdf transform}
\end{equation}%
In this case, there is generally no simplification in terms of the drift and
diffusion term, which take the form 
\begin{eqnarray}
\mu _{\bar{X}}\left( \bar{x};\theta \right)  &=&\mu _{X}\left(
F_{X}^{-1}\left( \bar{x};\theta \right) ;\theta \right) f_{X}\left(
F_{X}^{-1}\left( \bar{x};\theta \right) ;\theta \right)   \label{eq: drift 3}
\\
&&+\frac{1}{2}\sigma _{X}^{2}\left( F_{X}^{-1}\left( \bar{x};\theta \right)
;\theta \right) f_{X}^{\prime }\left( F_{X}^{-1}\left( \bar{x};\theta
\right) ;\theta \right)   \notag
\end{eqnarray}%
and%
\begin{equation}
\sigma _{\bar{X}}\left( \bar{x};\theta \right) =\sigma _{X}\left(
F_{X}^{-1}\left( \bar{x};\theta \right) ;\theta \right) f_{X}\left(
F_{X}^{-1}\left( \bar{x};\theta \right) ;\theta \right) .  \label{eq: diff 3}
\end{equation}%
for $\bar{x}\in \mathcal{\bar{X}}=\left( 0,1\right) $. But the marginal
distribution is now known with $\bar{X}_{t}\sim U\left( 0,1\right) $ and we
can directly identify the transformation function by $U\left( y\right)
=F_{Y}\left( y\right) $, c.f. eq. (\ref{U}). The identification condition
then takes the form:\bigskip 

\begin{description}
\item[Assumption 3.4.] With $\mu _{\bar{X}}\left( \bar{x};\theta \right) $
and $\sigma _{\bar{X}}\left( \bar{x};\theta \right) $ given in eqs. (\ref%
{eq: drift 3})-(\ref{eq: diff 3}), the following hold:%
\begin{equation*}
\forall \bar{x}\in \left( 0,1\right) :\mu _{\bar{X}}\left( \bar{x};\theta
\right) =\mu _{\bar{X}}\left( \bar{x};\tilde{\theta}\right) \text{ and }%
\sigma _{\bar{X}}\left( \bar{x};\theta \right) =\sigma _{\bar{X}}\left( \bar{%
x};\tilde{\theta}\right) \Leftrightarrow \theta =\tilde{\theta}.
\end{equation*}%
\bigskip
\end{description}

\begin{corollary}
\label{Cor: Scheme 3}Under Assumptions 2.1-2.2 and 3.1, $\mathcal{S}$ is
identified if and only if Assumption 3.4 is satisfied.
\end{corollary}

The above result is only useful for showing identification of a given UPD if 
$F^{-1}\left( \bar{x};\theta \right) $ is available on closed form. But
similar to the previous identification schemes, it demonstrates we can
restrict attention to diffusions with known marginal distributions in the
model building phase. Specifically, one can choose a known density $%
f_{X}\left( x\right) $ that describes the stationary distribution of $X$
together with a parametric specification for, say, the drift function. We
can then rearrange eq. (\ref{mpdf_x}) to back out the diffusion term of the
UPD:%
\begin{equation}
\sigma _{X}^{2}\left( x;\theta \right) =\frac{2}{f_{X}\left( x\right) }%
\int_{x_{l}}^{x}\mu _{X}\left( z;\theta \right) f_{X}\left( z\right) dz.
\label{eq: Model 1B}
\end{equation}%
If the drift is specified so that $\mu _{X}\left( \cdot ;\theta \right) \neq
\mu _{X}(\cdot ;\tilde{\theta})$ for $\theta \neq \tilde{\theta}$, then
Assumption 3.4 will be satisfied for this model. Alternatively, one could
choose a parametric specification of the diffusion term and then derive the
corresponding drift term of the UPD satisfying%
\begin{equation*}
\mu _{X}\left( x;\theta \right) =\frac{1}{2f_{X}\left( x\right) }\frac{%
\partial }{\partial x}\left[ \sigma _{X}^{2}\left( x;\theta \right)
f_{X}\left( x\right) \right] .
\end{equation*}%
The resulting copula diffusion model is identified as long as the chosen
diffusion term satisfies $\sigma _{X}\left( \cdot ;\theta \right) \neq
\sigma _{X}(\cdot ;\tilde{\theta})$ for $\theta \neq \tilde{\theta}$, then
Assumption 3.4 will be satisfied for this model.

Below, we apply the third identification scheme to the OU\ and CIR\
model:\bigskip

\noindent \textbf{Example 1 (continued). }The stationary distribution of (%
\ref{OU}) is $N\left( \alpha ,v^{2}\right) $ with $v^{2}=\sigma ^{2}/2\kappa 
$ and so the marginal density and cdf takes the form $f_{X}\left( x;\theta
\right) =\frac{1}{v}\phi \left( \frac{x-\alpha }{v}\right) $ and $%
F_{X}\left( x;\theta \right) =\Phi \left( \frac{x-\alpha }{v}\right) $,
where $\phi $ and $\Phi $ denote the density and cdf of the $N\left(
0,1\right) $ distribution. Applying the transformation (\ref{eq: cdf
transform}) yields, after some tedious calculations,%
\begin{equation*}
d\bar{X}_{t}=-2\kappa \Phi ^{-1}\left( \bar{X}_{t}\right) \phi \left( \Phi
^{-1}\left( \bar{X}_{t}\right) \right) dt+\sqrt{2\kappa }\phi \left( \Phi
^{-1}\left( \bar{X}_{t}\right) \right) dW_{t},
\end{equation*}%
which is independent of $\alpha $ and $\sigma ^{2}$ and these therefore have
to be fixed, leaving $\kappa $ as the only free parameter. This is the same
finding as with the first identification strategy.\bigskip

\noindent \textbf{Example 2 (continued). }The stationary distribution of the
CIR\ process is a $\Gamma $-distribution with scale parameter $\omega
=2\kappa /\sigma ^{2}$ and shape parameter $\nu =2\kappa \alpha /\sigma ^{2}$%
. Thus, the marginal density and cdf can be written as%
\begin{eqnarray*}
f_{X}\left( x;\theta \right) &=&f_{X}\left( x;\omega ,\nu \right) =\frac{%
\omega ^{\nu }}{\Gamma \left( \nu \right) }x^{\nu -1}e^{-\omega x} \\
F_{X}\left( x;\theta \right) &=&F_{X}\left( x;\omega ,\nu \right) =\frac{1}{%
\Gamma \left( \nu \right) }\gamma \left( \nu ,\omega x\right)
\end{eqnarray*}%
where $\Gamma \left( \nu \right) $ is the gamma function and $\gamma \left(
\nu ,\omega x\right) $ is the lower incomplete gamma function. Applying the
transformation (\ref{eq: cdf transform}) yields\textbf{\ }%
\begin{equation*}
\mu _{\bar{X}}\left( \bar{x};\theta \right) =\left[ \kappa \left( \frac{\nu 
}{2\kappa }-\frac{\gamma ^{-1}\left( \nu ,\bar{x}\Gamma \left( \nu \right)
\right) }{2\kappa }\right) +\left( \frac{\nu -1}{2}-\frac{\gamma ^{-1}\left(
\nu ,\bar{x}\Gamma \left( \nu \right) \right) }{2}\right) \right] \frac{%
2\kappa }{\Gamma \left( \nu \right) }\gamma ^{-1}\left( \nu ,\bar{x}\Gamma
\left( \nu \right) \right) ^{\nu -1}e^{-\gamma ^{-1}\left( \nu ,\bar{x}%
\Gamma \left( \nu \right) \right) }
\end{equation*}%
and%
\begin{equation*}
\sigma _{\bar{X}}^{2}\left( \bar{x};\theta \right) =2\kappa \gamma
^{-1}\left( \nu ,\bar{x}\Gamma \left( \nu \right) \right) \left[ \frac{1}{%
\Gamma \left( \nu \right) }\gamma ^{-1}\left( \nu ,\bar{x}\Gamma \left( \nu
\right) \right) ^{\nu -1}e^{-\gamma ^{-1}\left( \nu ,\bar{x}\Gamma \left(
\nu \right) \right) }\right] ^{2}.
\end{equation*}%
Note that $\mu _{\bar{X}}\left( \bar{x};\theta \right) $ and $\sigma _{\bar{X%
}}^{2}\left( \bar{x};\theta \right) $ only depend on $\kappa $ and $\nu $,
which means we can only identify $\alpha $ and $\sigma ^{2}$ up to a ratio
say $\alpha ^{\ast }=\alpha /\sigma ^{2}$. Hence, either $\alpha $ or $%
\sigma ^{2}$ must be fixed, which is in accordance with what we found when
applying the first identification strategy to the CIR. We could, for
example, set $\sigma ^{2}=2\kappa $ which leads to the following normalized
CIR%
\begin{equation*}
dX_{t}=\kappa \left( \alpha -X_{t}\right) dt+\sqrt{2\kappa X_{t}}dW_{t}.
\end{equation*}

\section{Estimation\label{Sec_Inference}}

In this section we develop two alternative semiparametric estimators of $%
\theta $ and $V$ for a given specification of the UPD. The first takes the
form of a two-step Pseudo Maximum Likelihood Estimator (PMLE). The second is
a semiparametric sieve-based ML estimator (SMLE). We consider two different
scenarios when developing estimators: In the first one (see Section \ref%
{sec: est low-freq}), $Y$ is observed at low frequency which we formally
define as the case when $\Delta >0$ is fixed as $n\rightarrow \infty $. In
the second one (see Section \ref{sec: est high-freq}), high-frequency data
is available so that $\Delta \rightarrow 0$ as $n\rightarrow \infty $.

\subsection{Low-frequency estimators\label{sec: est low-freq}}

To motivate the two estimators, suppose that $U$ is known, in which case the
MLE of $\theta $ is given by%
\begin{equation*}
\hat{\theta}_{\text{MLE}}=\arg \max_{\theta \in \Theta }L_{n}\left( \theta
,U\right) ,
\end{equation*}%
where $L_{n}\left( \theta ,U\right) $ is the log-likelihood of $\left\{
Y_{i\Delta }:i=0,1,...,n\right\} $,%
\begin{equation}
L_{n}\left( \theta ,U\right) =\frac{1}{n}\sum_{i=1}^{n}\left\{ \log
p_{X}\left( U\left( Y_{i\Delta }\right) |U\left( Y_{\left( i-1\right) \Delta
}\right) ;\theta \right) +\log U^{\prime }\left( Y_{i\Delta }\right)
\right\} ,  \label{eq: L(f) def}
\end{equation}%
where $p_{X}$ was is defined in eq. (\ref{eq: p_X def}). If $U$ is unknown,
the above estimator is not feasible and we instead have to estimate it
together with $\theta $.

Our PMLE assumes $Y$ is stationary in which case $U$ satisfies eq. (\ref{U}%
), where $F_{X}$ is known up to $\theta $ while $F_{Y}$ is unknown. The
latter can be estimated by the empirical cdf defined as%
\begin{equation*}
\tilde{F}_{Y}\left( y\right) =\frac{1}{n+1}\sum_{i=0}^{n}\mathbb{I}\left\{
Y_{i\Delta }\leq y\right\} ,
\end{equation*}%
where $\mathbb{I}\left\{ \cdot \right\} $ denotes the indicator function, or
alternatively by the following kernel smoothed empirical cdf, 
\begin{equation}
\hat{F}_{Y}\left( y\right) =\frac{1}{n+1}\sum_{i=0}^{n}\mathcal{K}_{h}\left(
Y_{i\Delta }-y\right) ,  \label{eq: kernel cdf}
\end{equation}%
where $\mathcal{K}_{h}\left( y\right) =\mathcal{K}\left( y/h\right) $ with $%
\mathcal{K}\left( y\right) =\int_{-\infty }^{y}K\left( z\right) dz$, $K$
being a kernel (e.g., the standard normal density), and $h>0$ a bandwidth.
Replacing $F_{Y}$ in eq. (\ref{U}) with either $\tilde{F}_{Y}$ or $\hat{F}%
_{Y}$, we obtain the following two alternative estimators of $U$,%
\begin{equation}
\tilde{U}\left( y;\theta \right) =F_{X}^{-1}(\tilde{F}_{Y}\left( y\right)
;\theta );\text{ \ \ }\hat{U}\left( y;\theta \right) =F_{X}^{-1}(\hat{F}%
_{Y}\left( y\right) ;\theta ).  \label{eq: U-tilde}
\end{equation}%
Since $\hat{F}_{Y}\left( y\right) =\tilde{F}_{Y}\left( y\right) +O\left(
h^{2}\right) $, the above two estimators of $U$ will be first-order
asymptotically equivalent under appropriate bandwidth conditions. A natural
way to estimate $\theta $ in our semiparametric framework would then be to
substitute either $\hat{U}\left( y;\theta \right) $ or $\tilde{U}\left(
y;\theta \right) $ into $L_{n}\left( \theta ,U\right) $. However, in the
latter case, this is not possible since $L_{n}\left( \theta ,U\right) $
depends on $U^{\prime }$ and $\tilde{U}$ is not differentiable. However,
note that%
\begin{equation}
U^{\prime }\left( y\right) =\frac{f_{Y}\left( y\right) }{f_{X}\left( U\left(
y\right) ;\theta \right) },  \label{U1}
\end{equation}%
so that $\log U^{\prime }\left( y\right) =\log f_{Y}\left( y\right) -\log
f_{X}\left( U\left( y\right) ;\theta \right) $. Since the first term is
parameter independent, it can be ignored and so we arrive at the following
semiparametric PMLE,%
\begin{equation*}
\hat{\theta}_{\text{PMLE}}=\arg \max_{\theta \in \Theta }\bar{L}_{n}(\theta ,%
\tilde{U}\left( \cdot ;\theta \right) ),
\end{equation*}%
where $\Theta $ is the parameter space and 
\begin{equation*}
\bar{L}_{n}\left( \theta ,U\right) =\frac{1}{n}\sum_{i=1}^{n}\left\{ \log
p_{X}\left( U\left( Y_{i\Delta }\right) |U\left( Y_{\left( i-1\right) \Delta
}\right) ;\theta \right) -\log f_{X}\left( U\left( Y_{i\Delta }\right)
;\theta \right) \right\}
\end{equation*}%
is $L_{n}\left( \theta ,U\right) -\sum_{i=1}^{n}\log f_{Y}\left( Y_{i\Delta
}\right) /n$. One can easily check that, by rewriting the above in terms of
the implied copula of $X$, this estimator is equivalent to the one analyzed
in Chen and Fan (2006).

Our second proposal, the SMLE, replaces the unknown density function $%
f_{Y}\left( y\right) $ by a sieve approximation $f_{Y,m}\left( y\right) \in 
\mathcal{F}_{m}$ where $\mathcal{F}_{m}$ is a finite-dimensional function
space reflecting the properties of $f_{Y}$, $m=1,2,...$. For a given
candidate density, we then compute 
\begin{equation*}
U\left( y;f_{Y,m},\theta \right) =F_{X}^{-1}(F_{Y,m}\left( y\right) ;\theta )
\end{equation*}%
where $F_{Y,m}\left( y\right) =\int_{y_{l}}^{y}f_{Y,m}\left( z\right) dz$.
Substituting this into the likelihood function yields the following
semiparametric sieve maximum-likelihood estimator,%
\begin{equation}
(\hat{\theta}_{\text{SMLE}},\hat{f}_{Y,m})=\arg \max_{\theta \in \Theta
,f_{Y,m}\in \mathcal{F}_{m}}L_{n}\left( \theta ,U\left( \cdot
;f_{Y,m},\theta \right) \right) .  \label{eq: SMLE 1}
\end{equation}%
The above SMLE is identical to the one proposed by Chen, Wu and Yi (2009)
for the estimation of copula-based Markov models, except that while they
estimate the parameters of a copula function, we estimate those of the drift
and diffusion functions of the UPD. In comparison with the PMLE, the
numerical implementation of the SMLE involves joint maximization over both $%
\theta $ and $\mathcal{F}_{m}$, which is a harder numerical problem and
potentially more time-consuming. In terms of statistical efficiency, $\hat{%
\theta}_{\text{SMLE}}$ will in general reach the semiparametric efficiency
bound under stationarity, while the PMLE is inefficient.

Both of the above estimators require us to evaluate $F_{X}^{-1}\left(
x;\theta \right) $ which in general is not available on closed form and so
has to be computed using numerical methods, e.g., numerical integration or
Monte Carlo methods combined with a equation solver. For the SMLE, one can
circumvent this issue by directly approximating $U$ instead of $f_{Y}$: For
a given finite-dimensional function space of one-to-one transformations $%
\mathcal{U}_{m}$, an alternative to the SMLE in (\ref{eq: SMLE 1}) is $(%
\tilde{\theta}_{\text{SMLE}},\tilde{U}_{m})=\arg \max_{\theta \in \Theta
,U_{m}\in \mathcal{U}_{m}}L_{n}\left( \theta ,U_{m}\right) $. We expect this
to be computationally more efficient compared to the density version above;
the theoretical analysis of this alternative SMLE is left for future
research.

Once an estimator for $\theta $ has been obtained, we can estimate the drift
and diffusion terms of $Y$ using the expressions given in (\ref{RDdrift})
and (\ref{RDdiffusion}) by replacing $\theta $ and $U$ with their
estimators. However, this involves estimating the first and second
derivative of $U$. For the SMLE this is not an issue assuming that $\mathcal{%
F}_{m}$ is a differentiable function space. For the PMLE, since $\tilde{U}%
\left( y;\theta \right) $ is not differentiable, we instead use the kernel
smoothed version $\hat{U}\left( y;\theta \right) $, leading to the following
three-step estimators of the drift and diffusion functions%
\begin{eqnarray}
\hat{\mu}_{Y}\left( y\right) &=&\frac{\mu _{X}(\hat{U}\left( y\right) ;\hat{%
\theta}_{\text{PMLE}})}{\hat{U}^{\prime }\left( y\right) }-\frac{1}{2}\sigma
_{X}^{2}(\hat{U}\left( y\right) ;\hat{\theta}_{\text{PMLE}})\frac{\hat{U}%
^{\prime \prime }\left( y\right) }{\hat{U}^{\prime }\left( y\right) ^{3}},
\label{eq: mu_Y est} \\
\hat{\sigma}_{Y}^{2}\left( y\right) &=&\frac{\sigma _{X}^{2}(\hat{U}\left(
y\right) ;\hat{\theta}_{\text{PMLE}})}{\hat{U}^{\prime }\left( y\right) ^{2}}%
,  \label{eq: sig_Y est}
\end{eqnarray}%
where $\hat{U}\left( y\right) =F_{X}^{-1}(\hat{F}_{Y}\left( y\right) ;\hat{%
\theta}_{\text{PMLE}})$.

\subsection{High-frequency estimators\label{sec: est high-freq}}

We now turn to the case where high-frequency data is available; this
scenario is formally modelled as $\Delta \rightarrow 0$ as $n\rightarrow
\infty $. The proposed estimators described in the previous section remains
valid, but an alternative estimation method is available in this case since
the exact density of the underlying UPD, $p_{X}$, is well-approximated by%
\begin{equation}
\hat{p}_{X}\left( x|x_{0};\theta \right) =\frac{1}{\sqrt{2\pi \Delta }}%
\sigma _{X}\left( x_{0};\theta \right) \exp \left[ -\frac{\left( x-x_{0}-\mu
_{X}\left( x_{0};\theta \right) \Delta \right) ^{2}}{2\sigma _{X}^{2}\left(
x_{0};\theta \right) \Delta }\right]  \label{eq: p_X approx}
\end{equation}%
as $\Delta \rightarrow 0$, c.f. Kessler (1997). We then propose to estimate $%
\theta $ using either the two-step or sieve approach described in the
previous section, except that we here replace $p_{X}\left( x|x_{0};\theta
\right) $ with its high-frequency approximation, $\hat{p}_{X}\left(
x|x_{0};\theta \right) $, in the definition of $L_{n}\left( \theta ,U\right) 
$ and $\bar{L}_{n}\left( \theta ,U\right) $. The advantage of doing so is
computational in that $\hat{p}_{X}\left( x|x_{0};\theta \right) $ is on
closed form for any given UPD while $p_{X}\left( x|x_{0};\theta \right) $
generally can only be evaluated using numerical methods as pointed out
earlier.

For most standard UPD's, the parameters can be decomposed into $\theta
=\left( \theta _{1},\theta _{2}\right) $ so that $\mu _{X}\left(
x_{0};\theta _{1}\right) =\mu _{X}\left( x_{0};\theta _{1}\right) $ and $%
\sigma _{X}\left( x_{0};\theta \right) =\sigma _{X}\left( x_{0};\theta
_{2}\right) $ only depends on the first and second component, respectively.
One could hope to be able to estimate $\theta _{1}$ and $\theta _{2}$
separately in this case. For known $U$, this is indeed possible. We could,
for example, use least-squares methods similar to Kanaya and Kristensen
(2018) where $\theta _{1}$ and $\theta _{2}$, respectively, are estimated by
the minimizers of the following two least-squares objectives,%
\begin{eqnarray}
L_{n,\Delta }^{\left( \mu \right) }\left( \theta _{1};U\right)
&=&\sum_{i=1}^{n}w_{i}^{\left( \mu \right) }\left( U\left( Y_{i\Delta
}\right) -U\left( Y_{\left( i-1\right) \Delta }\right) -\mu _{X}\left(
U\left( Y_{\left( i-1\right) \Delta }\right) ;\theta _{1}\right) \Delta
\right) ^{2},  \label{eq: LS-drift} \\
L_{n,\Delta }^{\left( \sigma \right) }\left( \theta _{2};U\right)
&=&\sum_{i=1}^{n}w_{i}^{\left( \sigma \right) }\left( \left\{ U\left(
Y_{i\Delta }\right) -U\left( Y_{\left( i-1\right) \Delta }\right) \right\}
^{2}-\sigma _{X}^{2}\left( U\left( Y_{\left( i-1\right) \Delta }\right)
;\theta _{2}\right) \Delta \right) ^{2},  \label{eq: LS-diff}
\end{eqnarray}%
where $w_{i}^{\left( \mu \right) }=w^{\left( \mu \right) }\left( Y_{\left(
i-1\right) \Delta },Y_{i\Delta }\right) $ and $w_{i}^{\left( \sigma \right)
}=w^{\left( \sigma \right) }\left( Y_{\left( i-1\right) \Delta },Y_{i\Delta
}\right) $ are weighting functions.

This approach, however, faces two complications in our setting: First, after
applying any of the three normalizations presented in Section \ref%
{Sec_Identification} in order to achieve identification, the resulting drift
and diffusion of the UPD tend to share parameters. Second, $U$ is unknown
and has to be estimated together with $\theta $. In the case of PMLE, $%
\tilde{U}\left( y;\theta \right) $ in eq. (\ref{eq: U-tilde}) generally
depends on both $\theta _{1}$ and $\theta _{2}$ since $f_{X}\left( x;\theta
\right) $ does. Thus, if we replace $U$ by $\tilde{U}\left( y;\theta \right) 
$ in the above objectives, we cannot separately estimate $\theta _{1}$ and $%
\theta _{2}$. Similarly, the SMLE requires joint estimation of $U$ together
with $\theta $ in which case it would have to be re-estimated for each of
the two objectives. In conclusion, these least-squares estimators are rarely
useful in practice.

Another alternative approach, inspired by Bandi and Phillips (2007), see
also Kristensen (2011), would be to first obtain non-parametric estimates of 
$\mu _{Y}$ and $\sigma _{Y}^{2}$ and then match these with the ones implied
by the copula model,%
\begin{equation*}
Q_{n,\Delta }^{\left( \mu \right) }\left( \mathcal{S}\right)
=\sum_{i=1}^{n}w_{i}^{\left( \mu \right) }\left( \hat{\mu}_{Y}\left(
Y_{i\Delta }\right) -\mu _{Y}\left( Y_{i\Delta };\mathcal{S}\right) \right)
^{2},\text{ \ \ }Q_{n,\Delta }^{\left( \sigma \right) }\left( \mathcal{S}%
\right) =\sum_{i=1}^{n}w_{i}^{\left( \sigma \right) }\left( \hat{\sigma}%
_{Y}^{2}\left( Y_{i\Delta }\right) -\sigma _{Y}^{2}\left( Y_{i\Delta };%
\mathcal{S}\right) \right) ^{2},
\end{equation*}%
where $\hat{\mu}_{Y}\left( \cdot \right) $ and $\hat{\sigma}_{Y}^{2}\left(
\cdot \right) $ are the first-step nonparametric estimators; see Bandi and
Phillips (2007) for their precise forms. This procedure suffers from the
same issue as the least-squares one described in the previous paragraph. An
additional complication is that it involves multiple smoothing parameters:
First, $\hat{\mu}_{Y}\left( \cdot \right) $ and $\hat{\sigma}_{Y}^{2}\left(
\cdot \right) $ depend on two bandwidths and converge with slow rates and,
second, $\mu _{Y}\left( \cdot ;\mathcal{S}\right) $ and $\sigma
_{Y}^{2}\left( \cdot ;\mathcal{S}\right) $ involve derivatives of $U$ and so
if we replace $U$ by its kernel-smoothed estimator, $\hat{U}$, the two
objective funtions will depend on the first and second order derivatives of
the kernel density estimator of $f_{Y}$, which in turn depends on additional
bandwidth. All together, these estimators will be complicated to implement
due to the multiple bandwidths that the econometrician have to choose.
Moreover, their asymptotic analysis and behaviour will be non-standard.

\section{Asymptotic Theory\label{sec: asymp theory}}

\subsection{Low-frequency Estimation of Parametric Component}

We here establish an asymptotic theory for the proposed estimators in the
case of low-frequency data ($\Delta >0$ fixed). In the theoretical analysis
we shall work under the following high-level identification condition:

\begin{description}
\item[Assumption 4.1] $\mathcal{S}_{0}$ is identified.\bigskip
\end{description}

The previous section provided three different sets of primitive conditions
for Assumption 4.1 to hold in terms of $\left( \mu _{Y}\left( \cdot ;%
\mathcal{S}\right) ,\sigma _{Y}\left( \cdot ;\mathcal{S}\right) \right) $.
This combined with Assumption 3.1 then implies that the mapping $\left( \mu
_{Y}\left( \cdot ;\mathcal{S}\right) ,\sigma _{Y}\left( \cdot ;\mathcal{S}%
\right) \right) \mapsto p_{Y}\left( y|y_{0};\mathcal{S}\right) $ is
injective so that different drift and diffusion terms lead to different
transition densities. One implication of Assumptions 3.1 and 4.1 is $E\left[
\log p_{Y}\left( Y_{\Delta }|Y_{0};\mathcal{S}\right) \right] <E\left[ \log
p_{Y}\left( Y_{\Delta }|Y_{0};\mathcal{S}_{0}\right) \right] $ for any $%
\mathcal{S}\neq \mathcal{S}_{0}$, c.f. Newey and McFadden (1994, Lemma 2.2).
This ensures that the SMLE identifies $\mathcal{S}_{0}$ in the limit.
Regarding the PMLE, we note that it replaces $U$ by $\hat{U}\left( y;\theta
\right) =F_{X}^{-1}(\hat{F}_{Y}\left( y;\theta \right) )$. By the LLN of
stationary and ergodic sequences, $\hat{U}\left( y;\theta \right)
\rightarrow ^{P}U\left( y;\theta \right) =F_{X}^{-1}\left( F_{Y}\left(
y;\theta \right) \right) $, where, by the same arguments as before, $E\left[
\log p_{Y}\left( Y_{\Delta }|Y_{0};\theta ,U\left( \cdot ;\theta \right)
\right) \right] <E\left[ \log p_{Y}\left( Y_{\Delta }|Y_{0};\theta
_{0},U\left( \cdot ;\theta _{0}\right) \right) \right] $. Thus, the PMLE
will also in the limit identify $\theta _{0}$.

Next, we import conditions from Chen et al. (2010) guaranteeing, in
conjunction with our own Assumptions 2.1-2.2, that the UPD\ $X$, and thereby 
$Y$, is stationary and $\beta $-mixing with mixing coefficients decaying at
either polynomial rate (c.f. Corollary 5.5 in Chen et al., 2010) or
geometric rate (c.f. Corollary 4.2 in Chen et al., 2010):

\begin{description}
\item[Assumption 4.2.] (i) $\mu _{X}$ and $\sigma _{X}^{2}$ satisfies%
\begin{equation*}
\lim_{x\rightarrow x_{r}}\left\{ \frac{\mu _{X}\left( x;\theta _{0}\right) }{%
\sigma _{X}\left( x;\theta _{0}\right) }-\frac{1}{2}\frac{\partial \sigma
_{X}\left( x;\theta _{0}\right) }{\partial x}\right\} \leq 0,\text{ \ \ }%
\lim_{x\rightarrow x_{u}}\left\{ \frac{\mu _{X}\left( x;\theta _{0}\right) }{%
\sigma _{X}\left( x;\theta _{0}\right) }-\frac{1}{2}\frac{\partial \sigma
_{X}\left( x;\theta _{0}\right) }{\partial x}\right\} \geq 0;
\end{equation*}%
(ii) With $s\left( x;\theta \right) $ and $S\left( x;\theta \right) $
defined in (\ref{eq: scale density}), 
\begin{equation*}
\lim_{x\rightarrow x_{r}}\left\{ \frac{s\left( x;\theta _{0}\right) \sigma
_{X}\left( x;\theta _{0}\right) }{S\left( x;\theta _{0}\right) }\right\} >0,%
\text{ \ \ }\lim_{x\rightarrow x_{u}}\left\{ \frac{s\left( x;\theta
_{0}\right) \sigma _{X}\left( x;\theta _{0}\right) }{S\left( x;\theta
_{0}\right) }\right\} <0;
\end{equation*}%
\bigskip
\end{description}

Assumption 4.2(ii) is a strengthening of Assumption 4.2(i). For the analysis
of the PMLE, Assumption 4.2(i) suffices while we need the stronger
Assumption 4.2(ii) to establish an asymptotic theory for the SMLE. As we
mentioned before, it is not always straightforward to verify the required
mixing conditions for copula-based (discrete-time) Markov models such as
Chen and Fan (2006) and Chen, Wu and Yi (2009). In contrast, either sets of
conditions stated in Assumption 4.2 can be easily verified by directly
examining the drift and diffusion functions of the UPD $X$.

Finally, we impose the same conditions as used in the asymptotic analysis of
the PMLE in Chen and Fan (2006) and Chen, Wu and Yi (2009), respectively, on
the copula implied by the chosen UPD and the sieve density in the case of
SMLE:

\begin{description}
\item[Assumption 4.3.] (i) $c_{X}\left( u_{0},u;\theta \right) $ defined in (%
\ref{dic}) satisfies the regularity conditions set out in Chen and Fan
(2006, A1-A3, A4 or A4', A5-A6); (ii) $c_{X}\left( u_{0},u;\theta \right) $
and the sieve space $\mathcal{F}_{m}$ satisfy Assumptions 3.1-3.4 and
4.1--4.7, respectively, in Chen, Wu and Yi (2009).\bigskip
\end{description}

We here abstain from stating the precise, mostly technical, conditions and
refer the interested reader to Chen and Fan (2006) and Chen, Wu and Yi
(2009); broadly speaking their conditions translate into moment bounds and
smoothness conditions on the log-transition density of the UPD. These
conditions depend on the precise choice of the UPD and so will have to be
verified on a case-by-case basis. In Appendix \ref{Sec: verification}, we
verify the conditions for models in Examples 1--2.

The following result now follows from the general theory of Chen and Fan
(2006) and Chen, Wu and Yi (2009), respectively:

\begin{theorem}
\label{theorem_AD_PMLE}Under Assumptions 2.1-2.2, 4.1, 4.2(i) and 4.3(i),%
\begin{equation*}
\sqrt{n}(\hat{\theta}_{\mathrm{PMLE}}-\theta _{0})\rightarrow ^{d}N\left(
0,B^{-1}\Sigma B^{-1}\right) ,
\end{equation*}%
where $B$ and $\Sigma $ are defined in Chen and Fan (2006, A1 and $%
A_{n}^{\ast }$).

Under Assumptions 2.1-2.2, 4.1, 4.2(ii) and 4.3(ii),%
\begin{equation*}
\sqrt{n}(\hat{\theta}_{\mathrm{SMLE}}-\theta _{0})\rightarrow ^{d}N\left( 0,%
\mathcal{I}_{\ast }^{-1}\left( \theta \right) \right) ,
\end{equation*}%
where $\mathcal{I}_{\ast }$ is defined in Chen, Wu and Yi (2009).
\end{theorem}

Consistent estimators of the asymptotic variances, $B^{-1}\Sigma B^{-1}$ and 
$\mathcal{I}_{\ast }^{-1}\left( \theta \right) $, can be found in Chen and
Fan (2006) and Chen, Wu and Yi (2009), respectively.

\subsection{High-frequency Estimation of Parametric Component}

Next, we discuss the asymptotic properties of the PMLE based on the
high-frequency log-likelihood that takes as input $\hat{p}_{X}\left(
x|x_{0};\theta \right) $ defined in eq. (\ref{eq: p_X approx}); a complete
analysis of the PMLE and SMLE\ in a high-frequency setting is left for
future research. In the following, we let $T:=n\Delta $ denote the sampling
range, which will be assumed to diverge as $\Delta \rightarrow 0$.

The high-frequency PMLE is given by $\hat{\theta}_{\text{PMLE}}=\arg
\max_{\theta \in \Theta }\hat{L}_{n}\left( \theta ,\tilde{U}\left( \cdot
;\theta \right) \right) $ where%
\begin{equation*}
\hat{L}_{n}\left( \theta ,U\right) =\frac{1}{n}\sum_{i=1}^{n}\left\{ \log 
\hat{p}_{X}\left( U\left( Y_{i\Delta }\right) |U\left( Y_{\left( i-1\right)
\Delta }\right) ;\theta \right) -\log f_{X}\left( U\left( Y_{i\Delta
}\right) ;\theta \right) \right\} ,
\end{equation*}%
and $\tilde{U}\left( Y_{i\Delta };\theta \right) $ defined in (\ref{eq:
U-tilde}). We first specialize the general result of Kanaya (2018, Theorem
2) by choosing $B=\psi =1$ and $K_{h}\left( y\right) =I\left\{ y\leq
0\right\} $ in his notation to obtain that under our Assumption 4.2,%
\begin{equation}
\sup_{y\in \mathcal{Y}}\left\vert \tilde{F}_{Y}\left( y\right) -F_{Y}\left(
y\right) \right\vert =O_{P}\left( \sqrt{\Delta }/\log \Delta \right)
+O_{P}\left( \log T/\sqrt{T}\right) ,  \label{eq: cdf conv}
\end{equation}%
where the two terms on the right-hand side correspond to discretization bias
and sampling variance, respectively. By letting $T$ grow sufficiently fast
as $\Delta \rightarrow 0$, the first term can be ignored. Under regularity
conditions on $\mu _{X}$ and $\sigma _{X}$ so that $\left( y,y_{0}\right)
\mapsto \hat{p}_{X}\left( F_{X}^{-1}\left( y;\theta \right)
|F_{Y}^{-1}\left( y_{0}\right) ;\theta \right) /f_{Y}\left( y_{0}\right) $
satisfies Lipschitz conditions similar to the ones in Chen and Fan (2006),
we then obtain%
\begin{equation*}
\sup_{\theta \in \theta }\left\vert \hat{L}_{n}\left( \theta ,\tilde{U}%
\left( \cdot ;\theta \right) \right) -\hat{L}_{n}\left( \theta ,U\left(
\cdot ;\theta \right) \right) \right\vert =O_{P}\left( \sqrt{\Delta }/\log
\Delta \right) +O_{P}\left( \log T/\sqrt{T}\right) ,
\end{equation*}%
where $U\left( y;\theta \right) =F_{Y}\left( F_{X}^{-1}\left( y;\theta
\right) \right) $. Consistency of the PMLE now follows by extending the
arguments of Kessler (1997) to allow for the presence of the
parameter-dependent transformation $U\left( y;\theta \right) $. Next, to
simplify our discussion of the asymptotic distribution of the PMLE, we
consider two special cases:

First, suppose that suppose that, after suitable normalizations, $\sigma
_{X}\left( x\right) $ is known and only $\mu _{X}\left( x;\theta \right) $
is parameter dependent. In this case, we expect that Kessler's results
generalize so that $\hat{\theta}_{\text{PMLE}}$ will converge with $\sqrt{T}$%
-rate towards a Normal distribution, where the asymptotic variance will have
to be adjusted to take into account the first-step estimation of $\hat{F}%
_{Y} $.

Next, consider the opposite scenario, $\mu _{X}\left( x_{0}\right) $ is
known and only $\sigma _{X}\left( x_{0};\theta \right) $ is parameter
dependent. With $U$ known, Kessler (1997) shows that $\hat{\theta}_{\text{%
PMLE}}$ converges with $\sqrt{n}$-rate towards a Normal distribution in this
case. Note the faster convergence rate compared to the drift estimator.
However, in our setting $U\left( y;\theta \right) $ is parameter dependent,
and as a consequence this result appears to no longer apply: $U\left(
y;\theta \right) $ enters $\hat{L}_{n}\left( \theta ,U\right) $ in the same
way that $\mu _{X}$ does and so the score of $\hat{L}_{n}\left( \theta
,U\left( \cdot ;\theta \right) \right) $ will have a component on the same
form as in the first case and so will converge with $\sqrt{T}$-rate instead
of $\sqrt{n}$-rate. Moreover, the presence of the first-step estimator $%
\tilde{F}_{Y}\left( y\right) $, which also converge with $\sqrt{T}$-rate,
will generate an additional variance term. In total, estimators of diffusion
parameters appear not to enjoy "super" consistency in our setting due to the
way that the unknown transformation $U$ enters the likelihood.

\subsection{Estimation of Drift and Diffusion Functions\label{SecSemiFuns}}

We here analyze the asymptotic properties of the kernel-based estimators of $%
\mu _{Y}$ and $\sigma _{Y}^{2}$ given in eqs. (\ref{eq: mu_Y est})-(\ref{eq:
sig_Y est}). We only do so for the low-frequency case; the analysis of the
high-frequency case should proceed in a similar fashion. Our analysis takes
as starting point the following regularity conditions on the estimator of
the parametric component and the kernel function:\bigskip

\begin{description}
\item[Assumption 4.4.] The transformation function $V$ is four times
continuously differentiable.

\item[Assumption 4.5.] The estimator $\hat{\theta}$ of the parameter of the
UPD $X$ is $\sqrt{n}$-consistent.

\item[Assumption 4.6.] The kernel $K$ is differentiable, and there exists
constants $D,\omega >0$ such that 
\begin{equation*}
\left\vert K^{\left( i\right) }\left( z\right) \right\vert \leq D\left\vert
z\right\vert ^{-\omega },\text{ \ \ }\left\vert K^{\left( i\right) }\left(
z\right) -K^{\left( i\right) }\left( \tilde{z}\right) \right\vert \leq
D\left\vert z-\tilde{z}\right\vert ,\text{ \ \ }i=0,1,
\end{equation*}%
where $K^{\left( i\right) }\left( z\right) $ denotes the $i$th derivative of 
$K\left( z\right) $. Moreover, $\int_{\mathbb{R}}K\left( z\right) dz=1$, $%
\int_{\mathbb{R}}zK\left( z\right) dz=0$ and $\kappa _{2}=\int_{\mathbb{R}%
}z^{2}K\left( z\right) dz<\infty $.\bigskip
\end{description}

Assumption 4.4 ensures the existence of the $3$rd and $4$th derivatives of $%
U\left( y\right) $, which in turn ensure that relevant quantities entering
the asymptotic distributions of $\hat{\mu}_{Y}$ and $\hat{\sigma}_{Y}^{2}$
are well defined. Assumption 4.5 implies that the asymptotic properties of $%
\hat{\mu}_{Y}$ and $\hat{\sigma}_{Y}^{2}$ are determined by the properties
of the kernel density estimator alone. The proposed PMLE and SMLE satisfy
this condition under our Assumptions 4.1-4.3, but other $\sqrt{n}$%
-consistent\ estimators are allowed for.\ Assumption 4.6 regulates the
kernel functions and allow for most standard kernels such as the Gaussian
and the Uniform kernels. Using the functional delta-method together with
standard results for kernel density estimators, as found in Robinson (1983),
we obtain:

\bigskip

\begin{theorem}
\label{AD_DD}Under Assumptions 2.1-2.2, 4.2(i), and 4.4-4.6, we have as $%
n\rightarrow \infty $, $h\rightarrow 0$ and $nh^{3}\rightarrow \infty $,%
\begin{equation*}
\sqrt{nh^{3}}\left\{ \hat{\mu}_{Y}\left( y\right) -\mu _{Y}\left( y\right)
-h^{2}B_{\mu _{Y}}\left( y\right) \right\} \rightarrow ^{d}N\left( 0,V_{\mu
_{Y}}\left( y\right) \right) ,
\end{equation*}%
where%
\begin{equation*}
B_{\mu _{Y}}\left( y\right) =-\frac{\kappa _{2}\sigma _{Y}^{2}\left(
y\right) f_{Y}^{\prime \prime \prime }\left( y\right) }{4f_{Y}\left(
y\right) },\text{ \ \ }V_{\mu _{Y}}\left( y\right) =\frac{\sigma
_{Y}^{4}\left( y\right) }{4f_{Y}\left( y\right) }\int_{\mathbb{R}}K^{\prime
}\left( z\right) ^{2}dz.
\end{equation*}%
Also, as $n\rightarrow \infty $, $h\rightarrow 0$ and $nh\rightarrow \infty $%
, we have%
\begin{equation*}
\sqrt{nh}\{\hat{\sigma}_{Y}^{2}\left( y\right) -\sigma _{Y}^{2}\left(
y\right) -h^{2}B_{\sigma _{Y}^{2}}\left( y\right) \}\rightarrow ^{d}N\left(
0,V_{\sigma ^{2}}\left( y\right) \right) ,
\end{equation*}%
where%
\begin{equation*}
B_{\sigma _{Y}^{2}}\left( y\right) =-\frac{\kappa _{2}\sigma _{Y}^{2}\left(
y\right) f_{Y}^{\prime \prime }\left( y\right) }{f_{Y}\left( y\right) },%
\text{ \ \ }V_{\sigma _{Y}^{2}}\left( y\right) =\frac{4\sigma _{Y}^{4}\left(
y\right) }{f_{Y}\left( y\right) }\int_{\mathbb{R}}K\left( z\right) ^{2}dz.
\end{equation*}
\end{theorem}

We see that both estimators suffer from smoothing biases, $B_{\mu
_{Y}}\left( y\right) $ and $B_{\sigma _{Y}^{2}}\left( y\right) $. If $%
h\rightarrow 0$ sufficiently fast, these biases will be negiglible. Also
note that the convergence rates of the drift estimator is slower compared to
the diffusion estimator. These features are similar to the asymptotic
properties of the semi-nonparametric drift and diffusion estimators
considered in Kristensen (2011).

\section{Monte\ Carlo Simulations\label{Sec_Simulation}}

In this section, we compare the finite sample performance of our
low-frequency semiparametric PMLE with that of a fully parametric PMLE
(described below) through Monte Carlo simulations.

\subsection{Data Generating Processes}

We consider the following normalized versions of the UPDs of Examples 1--2,%
\begin{eqnarray}
\text{OU} &:&dX_{t}=-\kappa X_{t}dt+\sqrt{2\kappa }dW_{t},\text{ \ \ }\theta
=\kappa ,  \label{Normalized_OU} \\
\text{CIR} &:&dX_{t}=\kappa \left( \alpha -X_{t}\right) dt+\sqrt{2\kappa
X_{t}}dW_{t},\text{ \ \ }\theta =\left( \kappa ,\alpha \right) .
\label{Normalized_CIR}
\end{eqnarray}%
The chosen normalizations have the advantage that the marginal distributions
of $X$ are invariant to the mean-reversion parameter $\kappa $. Hence, by
varying $\kappa $, we can change the persistence level of $X$ (and thus $Y$)
while keeping the marginal distributions fixed. In this way, we can examine
the impact of persistence on the performance of the proposed estimators of $%
\theta $, $\mu _{Y}$ and $\sigma _{Y}^{2}$.

Next, we specify the transformation of the DGP of $Y$. This is done by
choosing marginal cdf $F_{Y}\left( y;\phi \right) $, where $\phi $ is a
hyper parameter governing the shape of the cdf, which induces the
transformation $V\left( X_{t};\phi \right) =F_{Y}^{-1}\left( F_{X}\left(
X_{t};\theta \right) ;\phi \right) $. With $f_{Y}\left( y;\phi \right)
=F_{Y}^{\prime }\left( y;\phi \right) $, the transition density of the true
DGP of $Y$ then takes the form%
\begin{equation}
p_{Y}\left( y|y_{0};\theta ,\phi \right) =f_{Y}\left( y;\phi \right)
c_{X}\left( F_{Y}\left( y_{0};\phi \right) ,F_{Y}\left( y;\phi \right)
;\theta \right) .  \label{RDSKST_tpdf}
\end{equation}%
We choose $F_{Y}\left( y;\phi \right) $ as a flexible distribution to
reflect stylized features such as asymmetry and fat-tailedness of observed
financial data. Specifically, we use the Skewed Student-$t$\ (SKST)
Distribution of Hansen (1994) with density%
\begin{equation}
f_{Y}\left( y;\phi \right) =\left\{ 
\begin{array}{ccc}
\dfrac{bq}{v}\left( 1+\dfrac{1}{\tau -2}\left( \dfrac{\dfrac{b}{v}\left(
y-m\right) +a}{1-\lambda }\right) ^{2}\right) ^{-\left( \tau +1\right) /2} & 
\text{if} & y<m-av/b, \\ 
\dfrac{bq}{v}\left( 1+\dfrac{1}{\tau -2}\left( \dfrac{\dfrac{b}{v}\left(
y-m\right) +a}{1+\lambda }\right) ^{2}\right) ^{-\left( \tau +1\right) /2} & 
\text{if} & y\geq m-av/b,%
\end{array}%
\right.  \label{SKSTpdf}
\end{equation}%
where $v>0$, $2<\tau <\infty $, $-1<\lambda <1$, $a=4\lambda q\left( \dfrac{%
\tau -2}{\tau -1}\right) $, $b^{2}=1+3\lambda ^{2}-a^{2}$ and $q=\Gamma
\left( \left( \tau +1\right) /2\right) /\sqrt{\pi \left( \tau -2\right)
\Gamma ^{2}\left( \tau /2\right) }$. We collect the hyper parameters in $%
\phi =\left( m,v,\lambda ,\tau \right) $ which has to be chosen in order to
fully specify the DGP. While $m$ and $v$ are the unconditional mean and
standard deviation of the distribution, $\lambda $ controls the skewness and 
$\tau $ controls the degrees of freedom (hence the fat-tailedness) of the
distribution. The distribution reduces to the usual student-$t$ distribution
when $\lambda =0$. Due to its flexibility in modelling skewness and
kurtosis, the SKST distribution is often used in financial modelling. (c.f.
Patton, 2004; Jondeau and Rockinger, 2006; Bu, Fredj and Li, 2017).

The transformed diffusion $Y$ generated by the SKST marginal distribution
together with the normalized UPD in (\ref{Normalized_OU}) or (\ref%
{Normalized_CIR}) is referred to as the OU-SKST or the CIR-SKST model,
respectively. The true data-generating parameters $\phi $ and $\theta $ are
chosen as estimates obtained from fitting the parametric versions of the two
models to the 7-day Eurodollar interest rate time series\ used in A\"{\i}%
t-Sahalia (1996b). The estimation is based on a fully parametric two-stage
PMLE. In the first stage, the SKST distribution is fitted to the data (as if
they are i.i.d) to obtain $\hat{\phi}$. We then substitute $F_{Y}(y;\hat{\phi%
})$ and $f_{Y}(y;\hat{\phi})$\ into (\ref{RDSKST_tpdf}) which is then
maximized with respect to $\theta $ to obtain $\hat{\theta}$ for each of the
two UPD's. The calibrated parameter values of the marginal SKST distribution
are $(\hat{m},\hat{v},\hat{\lambda},\hat{\tau}%
)=(0.0835,0.0358,0.5193,25.3708)$, and those of the underlying OU and CIR\
diffusions are $\hat{\kappa}=1.1376$ and $\left( \hat{\kappa},\hat{\alpha}%
\right) =\left( 0.7653,1.1653\right) $, respectively.

We compare the fitted SKST and Normal distributions with a nonparametric
kernel estimate in Figure 1. We see that the SKST\ distribution does a
reasonable job at capturing the marginal distribution found in data while
the Normal one does not provide a very good fit.%
\begin{equation*}
\text{\lbrack Figure 1]}
\end{equation*}

Artificial samples of sizes $n=2202$ and $n=5505$, respectively, are then
generated using $\phi =\hat{\phi}$ and $\theta =\hat{\theta}$ as our true
data-generating parameters. For both OU-SKST and CIR-SKST, $\theta $
involves the mean-reversion parameter $\kappa $ which controls the level of
persistence. We create $3$ additional scenarios by multiplying $\kappa $ by
factors of 5, 10, and 20 while keeping everything else unchanged.
Collectively, we have a total of $8$ cases corresponding to $2$ sample sizes
and $4$ persistence levels. The maximum factor $20$ is chosen because the
implied 1st-order autocorrelation coefficient $\rho _{1}\approx 0.9$, which
is a reasonably high persistent level without being excessively close to the
unit root. Finally, 500 replications for each case are generated.

\subsection{Estimation Results}

We compare our low-frequency PMLE of $\theta $ with the corresponding fully
parametric PMLE (PPMLE) described above that we used for our calibration.
Note that the only difference between the two estimators is that the former
estimates the marginal distribution $F_{Y}$ parametrically, while the latter
estimates it nonparametrically.

The relative bias and RMSE (defined as the ratios of the actual bias and the
actual RMSE over the true parameter value, respectively) of the estimators
of the parametric components of the OU-SKST case are presented in Table 1.
Overall, the results from the two estimation methods are generally
comparable with the same magnitudes. The semiparametric PMLE tends to do
better in terms of bias while the parametric PMLE dominates in terms of
variance. However, as the level of persistence decreases, the two
estimators' performance is close to identical. 
\begin{equation*}
\text{\lbrack Table 1]}
\end{equation*}

The results for the CIR-SKST case are presented in Table 2 and 3 which are
qualitatively very similar to the ones for the OU-SKST. Overall, the
performance of the PMLE\ is comparable with that of the PPMLE with very
similar estimation errors. Moreover, the gap in the performance of the PMLE
relative to the PPMLE appears to narrow when the true DGP gets less
persistent.%
\begin{equation*}
\text{\lbrack Table 2 and 3]}
\end{equation*}

Next, we investigate the performance of the semiparametric estimators of $%
\mu _{Y}$ and $\sigma _{Y}^{2}$ in eqs. (\ref{eq: mu_Y est})-(\ref{eq: sig_Y
est}) relative to their fully parametric estimators. In Figure 2, we plot
their pointwise means and $95\%$ confidence bands from the 500 estimates
against the truth for the OU-SKST process with $\kappa =22.753$\ and sample
size 2202. First, it is worth noting that $\mu _{Y}$ and $\sigma _{Y}^{2}$
exhibit strong nonlinearities that closely resemble the nonlinearities
depicted in, for example, A\"{\i}t-Sahalia (1996b), Jiang and Knight (1997),
and Stanton (1997). Second, the mean estimates from both estimation methods
are fairly close to the truth, but the variability of the semiparametric
estimators is noticeably larger than the parametric ones, especially in the
right end of the range. This is not surprising: Firstly, as shown in Theorem %
\ref{AD_DD}, $\hat{\mu}_{Y}$ and $\hat{\sigma}_{Y}^{2}$ converges at slower
than $\sqrt{n}$-rate due to the use of kernel estimators of $f_{Y}$. From
Figure 1, we can see that $f_{Y}$ has a long right tail which is difficult
to estimate by the kernel estimator in small and moderate samples. Figure 3
presents the same estimators at sample size 5505. At this larger sample
size, the bias is even smaller for both methods and the variability of these
estimates are also reduced significantly. Overall, although the parametric
method obviously has the advantage due to its parametric structure, our
semiparametric method also provides fairly satisfactory estimation results.%
\begin{equation*}
\text{\lbrack Figure 2 and 3]}
\end{equation*}

The drift and diffusion estimators from the two methods where the true DGP
is the CIR-SKST process with $\kappa =15.307$\ and the two sample sizes are
presented in Figure 4 and 5, respectively. Almost identical qualitative
conclusions can be reached.%
\begin{equation*}
\text{\lbrack Figure 4 and 5]}
\end{equation*}

\section{Empirical Application\label{Sec_Application}}

\subsection{Data}

As an empirical illustration, we here model the time series dynamics of the
CBOE Volatility Index data using copula diffusion models. The data consists
of the daily VIX index from January 2, 1990 to July 19, 2019 ($7445$
observations). It is displayed and summarized in Figure 6 and Table 4,
respectively. The time series plot shows a clear pattern of mean reversion,
and Augmented Dickey-Fuller tests with reasonable lags all rejected the unit
root hypothesis at $5\%$ significance level, which justifies the use of
stationary diffusion models. The mean and the standard deviation is of VIX
is $19.21$ and $7.76$, respectively. Meanwhile, the skewness and the
kurtosis are $2.12$ and $10.85$, respectively, suggesting that the
stationary distribution deviates quite substantially from normality. This is
more formally confirmed by the highly significant Jarque-Bera test statistic
with a negligible $p$-value.%
\begin{equation*}
\text{\lbrack Figure 6 and Table 4]}
\end{equation*}

\subsection{Models}

We focus on whether two well known parametric transformed diffusion models
proposed for modelling VIX are supported by the data against their
semiparametric alternatives. The two parametric models are the
transformed-OU model of Detemple and Osakwe (2000) (DO) and the
transformed-CIR model of Eraker and Wang (2015) (EW). Specifically, the DO\
model is the exponential transform of the OU process, which can be written as%
\begin{equation*}
Y_{t}=\exp \left( X_{t}\right) ,\text{ \ \ \ \ \ }dX_{t}=\kappa \left(
\alpha -X_{t}\right) dt+\sigma dW_{t}
\end{equation*}%
and the EW model is a parameter-dependent transformation of the CIR process,
which is given by%
\begin{equation*}
Y_{t}=\frac{1}{X+\delta }+\varrho ,\text{ \ \ \ \ \ }dX_{t}=\kappa \left(
\alpha -X_{t}\right) dt+\sigma \sqrt{X_{t}}dW_{t}
\end{equation*}%
Meanwhile, the two semiparametric models we consider are the same two models
considered in our simulations, namely, the nonparametrically transformed OU
and CIR\ models, which we denote as NPTOU and NPTCIR, respectively. Their
associated normalized UPD processes are given in (\ref{Normalized_OU}) and (%
\ref{Normalized_CIR}). 

Importantly, we maintain the assumption that the VIX is a Markov diffusion
process. In particular, we rule out jumps and stochastic volatility (SV) in
the VIX which is inconsistent with the empirical findings of, e.g., Kaeck
and Alexander (2013). However, their models are fully parametric and so
impose much stronger functional form restrictions on the drift and diffusion
component compared to our semiparametric approach. Specifically, jumps and
SV components are often used to capture extremal events (fat tails). It is
possible that these components are needed in explaining the VIX dynamics due
to the restrictive drift and diffusion specifications they consider. Our
semiparametric approach allows for more flexibility in this respect and so
can be seen as a competing approach to capturing the same features in data.
An interesting research topic would be to develop tools that allow for
formal statistical comparison of our class of models against these
alternative ones.

\subsection{Results}

For each of the two UPDs, we examine whether the parametric specification of
the transformation is supported by the data. We do this by testing each of
the parametric models against the semiparametric alternative where the
transformation is left unspecified. We do so by computing a pseudo
Likelihood Ratio (pseudo-LR) test statistic defined as the difference
between the pseudo log-likelihood (pseudo-LL) of the semiparametric model
and the log-likelihood (LL) of the parametric model. Since the model under
the alternative is semiparametric and estimated by pseudo-ML, the pseudo-LL
test statistic will not follow a $\chi ^{2}$-distribution. We therefore
resort to a parametric bootstrap procedure:\ For each of the two pseudo-LR
test, we simulate $1000$ new time series from the parametric model using as
data-generating parameter values the MLEs obtained from the original sample.
For each of the $1000$ new data sets, of the same size as the original one,
we estimate both the parametric model and the semiparametric model and
compute the corresponding pseudo-LR statistic. Finally, we use the $95$th
and $99$th quantiles from the simulated distribution of the pseudo-LR
statistic as our $5\%$ and $1\%$ bootstrap critical values, respectively.

The pseudo-LL is computed using the log-likelihood given in (\ref{eq: L(f)
def}) with $U\left( y\right) $ and $\log U^{\prime }\left( y;\theta \right) $
replaced by $\tilde{U}^{\prime }\left( y;\theta \right) $ given in (\ref{eq:
U-tilde}) and $\log \tilde{U}^{\prime }\left( y;\theta \right) =\log \hat{f}%
_{Y}\left( y\right) -\log f_{X}\left( \tilde{U}\left( y\right) ;\theta
\right) $, respectively. Here, $\hat{f}_{Y}\left( y\right) $ is the kernel
density estimator which requires us choosing a bandwidth. There is a lack of
consensus on the right procedure for choosing bandwidths for kernel
estimators using dependent data. We therefore considered a sequence of
bandwidths constructed by multiplying the Silverman's rule of thumb
bandwidth, denoted as $h_{S}$, by a factor $k$ between $0.75$ and $1.75$ on
a small grid. Visual inspection of these density estimates revealed that
with $k$ is around $1.5$, the resulting density appears to be the most
satisfactory in terms of smoothness and the revelation of distributional
features of the data. For this reason, we report our inferential results
based on the relatively optimal bandwidth $1.5h_{S}=2.0730$ below. However,
our conclusions remain unchanged for any bandwidth within the aforementioned
range.

Our estimation and testing results are reported in Table 5. The upper panel
of the table presents the parameter estimates for the models together with
their standard errors in the parentheses underneath. For the two
semiparametric models, these were computed using the estimators proposed in
Chen and Fan (2006). Recall that due to normalization, only $\kappa $ is
estimated for the NPTOU model and only $\kappa $ and $\alpha $ for the
NPTCIR model.\ In addition, while $\kappa $ has the same interpretation
(i.e. rate of mean reversion) and scale in all four models, $\alpha $ has
different scales in the two transformed CIR models. For both the transformed
OU and the transformed CIR\ classes of models, we can see that the PMLEs of
the mean-reversion parameter $\hat{\kappa}$\ are slightly lower than their
corresponding MLE estimates. The same difference applies to their standard
errors. This shows that parametric (mis-)specification of the stationary
distribution does have a quite significant impact on the estimation of the
dynamic parameters.%
\begin{equation*}
\text{\lbrack Table 5]}
\end{equation*}

The lower panel presents the LL values and the our pseudo-LR test results.
We can see that the EW\ model has a much higher LL $(-1.1585)$ than the DO
model $\left( -1.1724\right) $, suggesting much better goodness of fit to
the data by the former. This is not entirely surprising because the EW model
is more flexible both in terms of the UPD and the transformation function
compared to the DO model. Meanwhile, the NPTCIR model has a higher pseudo-LL
than the NPTOU. Since they have identical stationary distributions, such a
difference is solely due to the additional flexibility of the UPD of the
former. Most importantly, we see that when the underlying diffusions are the
same, models with nonparametric transformation have much higher LLs than
those with parametric transformations. More specifically, the resulting
pseudo-LR between the NPTOU model and the DO model is $290.7263$, and that
between the NPTCIR model and the EW model is $40.8606$. This proves that the
exponential transformation of the DO model is too restrictive, and that
while the transformation function of the EW model is more flexible, it is
still rather restrictive relatively to our nonparametric alternative.

To formally assess the significance of the observed differences, we present
the empirical $5\%$ and $1\%$ critical values and the corresponding $p$%
-values of our pseudo-LR tests, obtained from our bootstrap procedure
described above. For both tests, we observe that those critical values are
all negative and the $p$-values are both exactly zero. This means that the
original pseudo-LRs of $290.7263$ and $40.8606$\ are not only far greater
than their corresponding empirical critical values but also greater than any
of the bootstrap pseudo-LRs when the parametric model under the null
hypothesis is true. This suggests that when either the DO model or the EW
model is the true model, the corresponding NPTOU model or the NPTCIR\ model
is unlikely to produce a higher LL value than the parametric model itself.
This is fairly strong evidence that the parametric assumptions made by the
DO and the EW models are not supported by our data and our nonparametrically
transformed models are strongly favored.

The reason for the rejection of the two parametric models can be found in
the implied stationary densities of the two models which we plot in Figure 7
together with the kernel density estimator. As can be seen from this figure,
the parametric specifications are unable to capture the middle range of the
empirical distribution of VIX; in contrast, the two semiparametric
alternatives are constructed so that they match the empirical distribution
exactly.

\begin{equation*}
\text{\lbrack Figure 7]}
\end{equation*}

\section{Conclusion\label{Sec_Conclusion}}

We propose a novel semiparametric approach for modelling stationary
nonlinear univariate diffusions. The class of models can be thought of as
Markov copula models where the copula is implied by the UPD model. Primitive
conditions for the identification of the UPD parameters together with the
unknown transformations from discrete samples are provided. We derive the
asymptotic properties for our semiparametric likelihood-based estimators of
the UPD parameters and kernel-based drift and diffusion estimators. Our
simulation results suggest that our semiparametric method performs well in
finite sample compared to the fully parametric method, and our relatively
simple application shows that the parametric assumptions on the
transformation function of the well known DO model and EW model are rejected
by the data against our nonparametric alternatives. Potential future work
under this framework may include extensions to multivariate diffusions and
jump-diffusions.

\newpage

%

\pagebreak

\appendix

\section{Proofs}

\begin{proof}[Proof of Theorem \protect\ref{Th: Master ID}]
From eqs. (\ref{eq: Y = V(tilde X)})-(\ref{eq: sig tilde X}), it is obvious
that (\ref{eq: equiv 1})-(\ref{eq: equiv 2}) imply $\mathcal{S}\sim \mathcal{%
\tilde{S}}$. Now, suppose that $\mathcal{S}\sim \mathcal{\tilde{S}}$; this
implies that $\mu _{Y}\left( y;\mathcal{S}\right) =\mu _{Y}\left( y;\mathcal{%
\tilde{S}}\right) $ and $\sigma _{Y}^{2}\left( y;\mathcal{S}\right) =\sigma
_{Y}^{2}\left( y;\mathcal{\tilde{S}}\right) $, where $\mu _{Y}$ and $\sigma
_{Y}^{2}$ are given in eqs. (\ref{RDdrift})-(\ref{RDdiffusion}). That is,
for all $y\in \mathcal{Y}$,%
\begin{eqnarray*}
\frac{\mu _{X}\left( U\left( y\right) ;\theta \right) }{U^{\prime }\left(
y\right) }-\frac{1}{2}\sigma _{X}^{2}\left( U\left( y\right) ;\theta \right) 
\frac{U^{\prime \prime }\left( y\right) }{U^{\prime }\left( y\right) ^{3}}
&=&\frac{\mu _{X}\left( \tilde{U}\left( y\right) ;\tilde{\theta}\right) }{%
\tilde{U}^{\prime }\left( y\right) }-\frac{1}{2}\sigma _{X}^{2}\left( \tilde{%
U}\left( y\right) ;\tilde{\theta}\right) \frac{\tilde{U}^{\prime \prime
}\left( y\right) }{\tilde{U}^{\prime }\left( y\right) ^{3}}, \\
\frac{\sigma _{X}\left( U\left( y\right) ;\theta \right) }{U^{\prime }\left(
y\right) } &=&\frac{\sigma _{X}\left( \tilde{U}\left( y\right) ;\tilde{\theta%
}\right) }{\tilde{U}^{\prime }\left( y\right) }.
\end{eqnarray*}%
Since $V$ is one-to-one we can set $y=V\left( x\right) $ in the above to
obtain the following for all $x\in \mathcal{X}$,%
\begin{eqnarray}
&&\frac{\mu _{X}\left( U\left( V\left( x\right) \right) ;\theta \right) }{%
U^{\prime }\left( V\left( x\right) \right) }-\frac{1}{2}\sigma
_{X}^{2}\left( U\left( V\left( x\right) \right) ;\theta \right) \frac{%
U^{\prime \prime }\left( V\left( x\right) \right) }{U^{\prime }\left(
V\left( x\right) \right) ^{3}}  \label{eq: equiv drift} \\
&=&\frac{\mu _{X}\left( \tilde{U}\left( V\left( x\right) \right) ;\tilde{%
\theta}\right) }{\tilde{U}^{\prime }\left( V\left( x\right) \right) }-\frac{1%
}{2}\sigma _{X}^{2}\left( \tilde{U}\left( V\left( x\right) \right) ;\tilde{%
\theta}\right) \frac{\tilde{U}^{\prime \prime }\left( V\left( x\right)
\right) }{\tilde{U}^{\prime }\left( V\left( x\right) \right) ^{3}},  \notag
\\
\frac{\sigma _{X}\left( U\left( V\left( x\right) \right) ;\theta \right) }{%
U^{\prime }\left( V\left( x\right) \right) } &=&\frac{\sigma _{X}\left( 
\tilde{U}\left( V\left( x\right) \right) ;\tilde{\theta}\right) }{\tilde{U}%
^{\prime }\left( V\left( x\right) \right) }.  \label{eq: equiv diff}
\end{eqnarray}%
Define $T\left( x\right) =\tilde{U}\left( V\left( x\right) \right)
\Leftrightarrow T^{-1}\left( x\right) =U\left( \tilde{V}\left( x\right)
\right) $, and observe that%
\begin{equation*}
U\left( V\left( x\right) \right) =x,\ \ U^{\prime }\left( V\left( x\right)
\right) V^{\prime }\left( x\right) =1,\text{ \ \ }\frac{\partial T\left(
x\right) }{\partial x}=\tilde{U}^{\prime }\left( V\left( x\right) \right)
V^{\prime }\left( x\right) .
\end{equation*}%
Eq. (\ref{eq: equiv diff}) combined with the above implies (\ref{eq: equiv 2}%
)(ii),%
\begin{equation}
\sigma _{X}\left( x;\theta \right) =\frac{\sigma _{X}\left( U\left( V\left(
x\right) \right) ;\theta \right) }{U^{\prime }\left( V\left( x\right)
\right) V^{\prime }\left( x\right) }=\frac{\sigma _{X}\left( \tilde{U}\left(
V\left( x\right) \right) ;\tilde{\theta}\right) }{\tilde{U}^{\prime }\left(
V\left( x\right) \right) V^{\prime }\left( x\right) }=\frac{\sigma
_{X}\left( T\left( x\right) ;\tilde{\theta}\right) }{\partial T\left(
x\right) /\left( \partial x\right) }=\sigma _{T^{-1}\left( X\right) }\left(
x;\tilde{\theta}\right) .  \label{eq: sig_X = sig_T}
\end{equation}

Next, divide through with $V^{\prime }\left( x\right) $ in (\ref{eq: equiv
drift}) and rearrange to obtain%
\begin{eqnarray*}
\mu _{X}\left( x;\theta \right)  &=&\frac{\mu _{X}\left( T\left( x\right) ;%
\tilde{\theta}\right) }{\partial T\left( x\right) /\left( \partial x\right) }%
+\frac{1}{2}\left\{ \sigma _{X}^{2}\left( x;\theta \right) \frac{U^{\prime
\prime }\left( V\left( x\right) \right) }{U^{\prime }\left( V\left( x\right)
\right) ^{3}V^{\prime }\left( x\right) }-\sigma _{X}^{2}\left( T^{-1}\left(
x\right) ;\tilde{\theta}\right) \frac{\tilde{U}^{\prime \prime }\left(
V\left( x\right) \right) }{\tilde{U}^{\prime }\left( V\left( x\right)
\right) ^{3}V^{\prime }\left( x\right) }\right\}  \\
&=&\frac{\mu _{X}\left( T\left( x\right) ;\tilde{\theta}\right) }{\partial
T\left( x\right) /\left( \partial x\right) }+\frac{1}{2}\sigma
_{X}^{2}\left( T\left( x\right) ;\tilde{\theta}\right) \left\{ \frac{1}{%
\tilde{U}^{\prime }\left( V\left( x\right) \right) ^{2}V^{\prime }\left(
x\right) ^{3}}\frac{U^{\prime \prime }\left( V\left( x\right) \right) }{%
U^{\prime }\left( V\left( x\right) \right) ^{3}}-\frac{\tilde{U}^{\prime
\prime }\left( V\left( x\right) \right) }{\tilde{U}^{\prime }\left( V\left(
x\right) \right) ^{3}V^{\prime }\left( x\right) }\right\} 
\end{eqnarray*}%
where the second equality uses (\ref{eq: sig_X = sig_T}). Eq. (\ref{eq:
equiv 2})(i) now follows since%
\begin{eqnarray*}
&&\frac{1}{\tilde{U}^{\prime }\left( V\left( x\right) \right) ^{2}V^{\prime
}\left( x\right) ^{3}}\frac{U^{\prime \prime }\left( V\left( x\right)
\right) }{U^{\prime }\left( V\left( x\right) \right) ^{3}}-\frac{\tilde{U}%
^{\prime \prime }\left( V\left( x\right) \right) }{\tilde{U}^{\prime }\left(
V\left( x\right) \right) ^{3}V^{\prime }\left( x\right) } \\
&=&\frac{1}{\tilde{U}^{\prime }\left( V\left( x\right) \right) ^{3}V^{\prime
}\left( x\right) ^{3}}\left[ \frac{\tilde{U}^{\prime }\left( V\left(
x\right) \right) U^{\prime \prime }\left( V\left( x\right) \right) }{%
U^{\prime }\left( V\left( x\right) \right) ^{3}}-\tilde{U}^{\prime \prime
}\left( V\left( x\right) \right) V^{\prime }\left( x\right) ^{2}\right]  \\
&=&\frac{-1}{\tilde{U}^{\prime }\left( V\left( x\right) \right)
^{3}V^{\prime }\left( x\right) ^{3}}\left[ \tilde{U}^{\prime }\left( V\left(
x\right) \right) V^{\prime \prime }\left( x\right) +\tilde{U}^{\prime \prime
}\left( V\left( x\right) \right) V^{\prime }\left( x\right) ^{2}\right]  \\
&=&-\frac{\partial ^{2}T\left( x\right) /\left( \partial x^{2}\right) }{%
\partial T\left( x\right) /\left( \partial x\right) ^{3}}.
\end{eqnarray*}
\end{proof}

\bigskip

\begin{proof}[Proof of Theorem \protect\ref{theorem_AD_PMLE}]
We first note that the PMLE takes the same form as the one analyzed in Chen
and Fan (2006) with the general copula considered in their work satisfying
eq. (\ref{dic}). The desired result will follow if we can verify that the
conditions stated in their proof are satisfied by our assumptions: First, by
Assumptions 2.1, the discrete sample $\left\{ X_{i\Delta }:i=0,1,\ldots
,n\right\} $ generated by the UPD$\ X$ is first-order Markovian and with
marginal density $f_{X}\left( x;\theta \right) $\ and transition density $%
p_{X}\left( x|x_{0};\theta \right) $. Hence, the copula density $c_{X}\left(
u_{0},u;\theta \right) $\ in (\ref{dic}) implied by $X$ is absolutely
continuous with respect to the Lebesgue measure on $\left[ 0,1\right] ^{2}$
due to its continuity in $F_{X}\left( x;\theta \right) $, $f_{X}\left(
x;\theta \right) $\ and $p_{X}\left( x|x_{0};\theta \right) $. Moreover, the
implied copula is neither the Fr\'{e}chet-Hoeffding upper or lower bound due
to Assumption 2.1, i.e., $\sigma _{X}^{2}\left( x;\theta \right) >0$ for all 
$x\in \mathcal{X}$. Thus, Chen and Fan (2006, Assumption 1) is satisfied.
Second, our Assumption 4.2(i) ensures that $X$ is $\beta $-mixing with
polynomial\ decay rate. Third, by Theorem \ref{Th: Y prop}, $Y$ is mixing
with the same mixing properties as $X$ and so satisfies Chen and Fan (2006,
Assumption 1). The remaining conditions are met by Assumption 4.3(i).

For the analysis of the proposed sieve MLE, we note that it takes the same
form as the one analyzed in Chen, Wu and Yi (2009) and so their results
carry over to our setting. Their Assumption M and assumption of $\beta $%
-mixing property are satisfied by $Y$ under our Assumptions 2.1, 2.2, and
4.2(ii) together with our Theorem \ref{Th: Y prop}. The remaining conditions
are met by Assumption 4.3(ii).
\end{proof}

\bigskip

\begin{proof}[Proof of Theorem \protect\ref{AD_DD}]
Similar to the proof strategy employed in Lemma \ref{Lem: U deriv}, we define%
\begin{equation*}
\tilde{\mu}_{Y}\left( y\right) =\frac{\mu _{X}\left( U\left( y\right)
;\theta \right) }{U^{\prime }\left( y\right) }-\frac{1}{2}\sigma
_{X}^{2}(U\left( y\right) ;\theta )\frac{\hat{U}^{\prime \prime }\left(
y\right) }{U^{\prime }\left( y\right) ^{3}},\text{ \ \ }\tilde{\sigma}%
_{Y}^{2}\left( y\right) =\frac{\sigma _{X}^{2}\left( U\left( y\right)
;\theta \right) }{\hat{U}^{\prime }\left( y\right) ^{2}},
\end{equation*}%
and, with $f_{Y}^{\left( i\right) }$ denoting the $i$th derivative of $f_{Y}$
and similar for other functions, arrive at%
\begin{eqnarray*}
&&\sqrt{nh^{3}}\left\{ \hat{\mu}_{Y}\left( y\right) -\mu _{Y}\left( y\right)
-\frac{1}{2}h^{2}\kappa _{2}\frac{f_{Y}^{\left( 3\right) }\left( y\right) }{%
f_{X}\left( U\left( y\right) ;\theta \right) }\left[ -\frac{\sigma
_{X}^{2}\left( U\left( y\right) ;\theta \right) }{2U^{\prime }\left(
y\right) ^{3}}\right] \right\}  \\
&=&\sqrt{nh^{3}}\left\{ \tilde{\mu}_{Y}\left( y\right) -\mu _{Y}\left(
y\right) -\frac{1}{2}h^{2}\kappa _{2}\frac{f_{Y}^{\left( 3\right) }\left(
y\right) }{f_{X}\left( U\left( y\right) ;\theta \right) }\left[ -\frac{%
\sigma _{X}^{2}\left( U\left( y\right) ;\theta \right) }{2U^{\prime }\left(
y\right) ^{3}}\right] \right\} +o_{p}\left( 1\right)  \\
&=&-\frac{\sigma _{X}^{2}\left( U\left( y\right) ;\theta \right) }{%
2U^{\prime }\left( y\right) ^{3}}\sqrt{nh^{3}}\left\{ \hat{U}^{\left(
2\right) }\left( y\right) -U^{\left( 2\right) }\left( y\right) -\frac{1}{2}%
h^{2}\kappa _{2}\frac{f_{Y}^{\left( 3\right) }\left( y\right) }{f_{X}\left(
U\left( y\right) ;\theta \right) }\right\} +o_{p}\left( 1\right) ,
\end{eqnarray*}%
and 
\begin{eqnarray*}
&&\sqrt{nh}\left\{ \hat{\sigma}_{Y}^{2}\left( y\right) -\sigma
_{Y}^{2}\left( y\right) -\frac{1}{2}h^{2}\kappa _{2}\frac{f_{Y}^{\left(
2\right) }\left( y\right) }{f_{X}\left( U\left( y\right) ;\theta \right) }%
\left[ -\frac{2\sigma _{X}^{2}\left( U\left( y\right) ;\theta \right) }{%
U^{\prime }\left( y\right) ^{3}}\right] \right\}  \\
&=&\sqrt{nh}\left\{ \tilde{\sigma}_{Y}^{2}\left( y\right) -\sigma
_{Y}^{2}\left( y\right) -\frac{1}{2}h^{2}\kappa _{2}\frac{f_{Y}^{\left(
2\right) }\left( y\right) }{f_{X}\left( U\left( y\right) ;\theta \right) }%
\left[ -\frac{2\sigma _{X}^{2}\left( U\left( y\right) ;\theta \right) }{%
U^{\prime }\left( y\right) ^{3}}\right] \right\} +o_{p}\left( 1\right)  \\
&=&-\frac{2\sigma _{X}^{2}\left( U\left( y\right) ;\theta \right) }{%
U^{\prime }\left( y\right) ^{3}}\sqrt{nh}\left\{ \hat{U}^{\prime }\left(
y\right) -U^{\prime }\left( y\right) -\frac{1}{2}h^{2}\kappa _{2}\frac{%
f_{Y}^{\left( 2\right) }\left( y\right) }{f_{X}\left( U\left( y\right)
;\theta \right) }\right\} +o_{p}\left( 1\right) .
\end{eqnarray*}%
These together with (\ref{AD_U1}) and (\ref{AD_U2}) of Lemma \ref{Lem: U
deriv} and Slutsky's Theorem complete the proof.
\end{proof}

\section{Verification of conditions for OU\ and CIR\ model\label{Sec:
verification}}

We here verify the technical conditions of Chen and Fan (2006) for the
normalized versions of the OU and CIR\ model given in eqs. (\ref%
{Normalized_OU}) and (\ref{Normalized_CIR}), respectively. For both
examples, we will require that $U\left( y;\theta \right) $, as defined in
eq. (\ref{U}), and its first and second-order derivatives w.r.t $\theta $
are polynomially bounded in $y$. This imposes growth restrictions on the
transformation function and is used to easily verify various moment
conditions in the following. Also note that the criterion $l\left(
U_{i-1},U_{i};\theta \right) $ in Chen and Fan (2006) takes the form $%
l\left( U_{i-1},U_{i};\theta \right) :=\log p_{X}\left( U\left( Y_{i\Delta
};\theta \right) ;U\left( Y_{\left( i-1\right) \Delta };\theta \right)
;\theta \right) -\log f_{X}\left( U\left( Y_{i\Delta };\theta \right)
;\theta \right) $, where $U_{i}=F_{Y}\left( Y_{i\Delta }\right) $, in our
notation.

\subsection{OU model}

\noindent \textbf{Assumption 4.2:} It is easily seen that $\left\{ \frac{\mu
_{X}\left( x;\theta _{0}\right) }{\sigma _{X}\left( x;\theta _{0}\right) }-%
\frac{1}{2}\frac{\partial \sigma _{X}\left( x;\theta _{0}\right) }{\partial x%
}\right\} =-\sqrt{\frac{\kappa }{2}}x$ and $\frac{s\left( x;\theta
_{0}\right) \sigma _{X}\left( x;\theta _{0}\right) }{S\left( x;\theta
_{0}\right) }=\exp \left( \frac{x^{2}}{2}\right) /\int_{x^{\ast }}^{x}\exp
\left( \frac{z^{2}}{2}\right) dz$. Assumption 4.2 is verified by taking the
relevant limits.\smallskip 

\noindent \textbf{Assumption 4.3: }The implied copula of the normalized OU
process is Gaussian, for which Assumption 4.3(i) and 4.3(ii) are satisfied
as discussed in Chen and Fan (2006) and Chen, Wu, and Yi (2009),
respectively.

\subsection{CIR\ model}

\noindent \textbf{Assumption 4.2: }We obtain $\left\{ \frac{\mu _{X}\left(
x;\theta _{0}\right) }{\sigma _{X}\left( x;\theta _{0}\right) }-\frac{1}{2}%
\frac{\partial \sigma _{X}\left( x;\theta _{0}\right) }{\partial x}\right\} =%
\frac{\left( 2\alpha -1\right) }{2}\sqrt{\frac{\kappa }{2x}}-\sqrt{\frac{%
\kappa }{4x}}$ and $\frac{s\left( x;\theta _{0}\right) \sigma _{X}\left(
x;\theta _{0}\right) }{S\left( x;\theta _{0}\right) }=\frac{\exp \left\{
x\right\} }{x^{\alpha }}\sqrt{2\kappa }\sqrt{x}/\int_{x^{\ast }}^{x}\frac{%
\exp \left\{ z\right\} }{z^{\alpha }}dz$ and the assumption is verified by
taking relevant limits.\smallskip 

\noindent \textbf{Assumption 4.3. }First observe that%
\begin{equation*}
p_{X}\left( x|x_{0};\theta \right) =\exp \left[ c_{0}\left( \theta \right)
-c\left( \theta \right) \left( x+e^{-\kappa \Delta }x_{0}\right) \right] 
\frac{x_{0}}{x}I_{\alpha -1}\left( 2c^{2}\left( \theta \right) \sqrt{xx_{0}}%
\right) ,
\end{equation*}%
where $I_{q}\left( \cdot \right) $ is the so-called modified Bessel function
of the first kind and of order $q$ and $c_{0}\left( \theta ,\Delta \right) >0
$ and $c\left( \theta ,\Delta \right) >0$ are analytic functions. Moreover, $%
f_{X}$ is here the density of a gamma distribution and so all polynomial
moments of $X$ exist. Since $U$ is assumed to be polynomially bounded, this
implies that all polynomial moments of $Y$ also exist. All smoothness
conditions imposed in Chen and Fan (2006) are trivially satisfied since $%
p_{X}\left( x|x_{0};\theta \right) $ and $U\left( y;\theta \right) $ are
twice continuously differentiable w.r.t their arguments and so will not be
discussed any further. Similarly, we have already shown that $Y$ is
geometrically mixing. It remains to verify the moment conditions and the
identifying restrictions imposed in C1-C.5 in Proposition 4.2 and A2-A6 in
Chen and Fan (2006).\smallskip 

\noindent \textbf{C1} is satisfied if we restrict $\theta =\left( \alpha
,\kappa \right) $ to be situated in a compact set on $\mathbb{R}_{+}^{2}$
that contains the true value. Observe that%
\begin{equation*}
\log p_{X}\left( x|x_{0};\theta \right) =c_{0}\left( \theta \right) -c\left(
\theta \right) \left( x+e^{-\kappa \Delta }x_{0}\right) +\log \left( \frac{%
x_{0}}{x}\right) +\log I_{\alpha -1}\left( 2c^{2}\left( \theta \right) \sqrt{%
xx_{0}}\right) .
\end{equation*}%
Thus,%
\begin{eqnarray*}
s_{\theta }\left( x|x_{0};\theta \right)  &:&=\frac{\partial \log
p_{X}\left( x;x_{0};\theta \right) }{\partial \theta } \\
&=&\dot{c}_{0}\left( \theta \right) -\dot{c}\left( \theta \right) \left(
x+e^{-\kappa \Delta }x_{0}\right) +c\left( \theta \right) \Delta e^{-\kappa
\Delta }x_{0}+\frac{I_{\alpha -1}^{\prime }\left( 2c^{2}\left( \theta
\right) \sqrt{xx_{0}}\right) 4c\left( \theta \right) \sqrt{xx_{0}}\dot{c}%
\left( \theta \right) }{I_{\alpha -1}\left( 2c^{2}\left( \theta \right) 
\sqrt{xx_{0}}\right) } \\
&&+\left[ 
\begin{array}{c}
\frac{\dot{I}_{\alpha -1}\left( 2c^{2}\left( \theta \right) \sqrt{xx_{0}}%
\right) }{I_{\alpha -1}\left( 2c^{2}\left( \theta \right) \sqrt{xx_{0}}%
\right) } \\ 
0%
\end{array}%
\right] ,
\end{eqnarray*}%
where $\dot{c}_{0}\left( \theta \right) =\partial c_{0}\left( \theta \right)
/\left( \partial \theta \right) $ and similar for other functions, $%
I_{\alpha -1}^{\prime }\left( x\right) =\partial I_{\alpha -1}\left(
x\right) /\left( \partial x\right) $, and $\dot{I}_{\alpha -1}\left(
x\right) =\partial I_{\alpha -1}\left( x\right) /\left( \partial \alpha
\right) $. It is easily verified that $\left\vert I_{\alpha -1}^{\prime
}\left( x\right) /I_{\alpha -1}\left( x\right) \right\vert $ and $\left\vert
\left\vert I_{\alpha -1}^{\prime }\left( x\right) /I_{\alpha -1}\left(
x\right) \right\vert \right\vert $ are both bounded by a polynomial in $x$.
Thus, $\left\Vert s_{X}\left( x|x_{0};\theta \right) \right\Vert $ is
bounded by a polynomial uniformly in $\theta \in \Theta $. The expressions
of $s_{x}\left( x|x_{0};\theta \right) :=\partial \log p_{X}\left(
x;x_{0};\theta \right) /\left( \partial x\right) $ and $s_{x_{0}}\left(
x|x_{0};\theta \right) :=\partial \log p_{X}\left( x;x_{0};\theta \right)
/\left( \partial x_{0}\right) $ are on a similar form and also polynomially
bounded. Now, observe that%
\begin{eqnarray*}
l_{\theta }\left( U_{i-1},U_{i};\theta \right)  &:&=\frac{\partial l\left(
U_{i-1},U_{i};\theta \right) }{\partial \theta } \\
&=&s_{\theta }\left( U\left( Y_{i\Delta };\theta \right) |U\left( Y_{\left(
i-1\right) \Delta };\theta \right) ;\theta \right)  \\
&&+s_{x}\left( U\left( Y_{i\Delta };\theta \right) |U\left( Y_{\left(
i-1\right) \Delta };\theta \right) ;\theta \right) \dot{U}\left( Y_{i\Delta
};\theta \right)  \\
&&+s_{x_{0}}\left( U\left( Y_{i\Delta };\theta \right) |U\left( Y_{\left(
i-1\right) \Delta };\theta \right) ;\theta \right) \dot{U}\left( Y_{\left(
i-1\right) \Delta };\theta \right)  \\
&&-\frac{\partial \log f_{X}\left( U\left( Y_{i\Delta };\theta \right)
;\theta \right) }{\partial \theta }.
\end{eqnarray*}%
Given that the model is correctly specified and identified, it follows by
standard arguments for MLE that $E\left[ l_{\theta }\left(
U_{i},U_{i-1};\theta \right) \right] =0$ if and only if $\theta $ equals the
true value.\smallskip 

\noindent \textbf{C4}. From the above expression of\ $l_{\theta }\left(
U_{i},U_{i-1};\theta \right) $ together with our assumption on $U\left(
y;\theta \right) $, it is easily checked that it is bounded by a polynomial
in\ $\left( Y_{i\Delta },Y_{\left( i-1\right) \Delta }\right) $ uniformly in 
$\theta \in \Theta $. It now follows that $E\left[ \sup_{\theta }\left\Vert
l_{\theta }\left( U_{i},U_{i-1};\theta \right) \right\Vert ^{p}\right]
<\infty $ for any $p\geq 1$.\smallskip 

\noindent \textbf{C5.} 
\begin{equation*}
l_{\theta ,1}\left( U_{i-1},U_{i};\theta \right) =\frac{\partial l_{\theta
}\left( U_{i-1},U_{i};\theta \right) }{\partial U_{i-1}},\text{ \ \ }%
l_{\theta ,2}\left( U_{i-1},U_{i};\theta \right) =\frac{\partial l_{\theta
}\left( U_{i-1},U_{i};\theta \right) }{\partial U_{i}}
\end{equation*}%
are again bounded by polynomials in\ $\left( Y_{i\Delta },Y_{\left(
i-1\right) \Delta }\right) $ and so have all relevant moments.\smallskip 

\noindent \textbf{A1(ii)-(iii).} With $W_{1,i}$ and $W_{2,i}$ defined in
(4.2)-(4.3) in Chen and Fan (2006) and 
\begin{equation*}
l_{\theta ,\theta }\left( U_{i-1},U_{i};\theta \right) =\frac{\partial
^{2}l\left( U_{i-1},U_{i};\theta \right) }{\partial \theta \partial \theta
^{\prime }},
\end{equation*}%
\begin{equation*}
\lim_{n\rightarrow \infty }\text{Var}\left( \frac{1}{\sqrt{n}}%
\sum_{i=1}^{n}\left\{ l_{\theta }\left( U_{i-1},U_{i};\theta \right)
+W_{1,i}+W_{2,i}\right\} \right) ,
\end{equation*}%
and $E\left[ l_{\theta ,\theta }\left( U_{i-1},U_{i};\theta \right) \right] $
to have full rank. We have been unable to verify these two conditions due to
the complex form of the score and hessian of the CIR\ model.\smallskip 

\noindent \textbf{A4.} Observe that $\left\vert W_{1,i}\right\vert \leq E%
\left[ \left\vert U_{i-1}\right\vert \left\Vert l_{\theta ,1}\left(
U_{i-1},U_{i};\theta \right) \right\Vert \right] <\infty $ and similar for $%
W_{2,i}$. Thus, both have all relevant moments.\smallskip 

\noindent \textbf{A5-A6} have already been verified above.

\section{Lemma}

\begin{lemma}
\label{Lem: U deriv}\textit{Under Assumptions 2.1-2.2, 4.2(i), and 4.4-4.6,
we have as }$n\rightarrow \infty $\textit{, }$h\rightarrow 0$\textit{, }$%
nh\rightarrow \infty $\textit{,}%
\begin{equation}
\sqrt{nh}\left\{ \hat{U}^{\prime }\left( y\right) -U^{\prime }\left(
y\right) -\frac{1}{2}h^{2}\kappa _{2}\frac{f_{Y}^{\prime \prime }\left(
y\right) }{f_{X}\left( U\left( y\right) ;\theta _{0}\right) }\right\}
\rightarrow ^{d}N\left( 0,\frac{U^{\prime }\left( y\right) ^{2}}{f_{Y}\left(
y\right) }\int_{\mathbb{R}}K\left( z\right) ^{2}dz\right) ,  \label{AD_U1}
\end{equation}%
\textit{and as }$n\rightarrow \infty $\textit{, }$h\rightarrow 0$\textit{, }$%
nh^{3}\rightarrow \infty $\textit{,}%
\begin{equation}
\sqrt{nh^{3}}\left\{ \hat{U}^{\prime \prime }\left( y\right) -U^{\prime
\prime }\left( y\right) -\frac{1}{2}h^{2}\kappa _{2}\frac{f_{Y}^{\prime
\prime \prime }\left( y\right) }{f_{X}\left( U\left( y\right) ;\theta
_{0}\right) }\right\} \rightarrow ^{d}N\left( 0,\frac{U^{\prime }\left(
y\right) ^{2}}{f_{Y}\left( y\right) }\int_{\mathbb{R}}K^{\prime }\left(
z\right) ^{2}dz\right) .  \label{AD_U2}
\end{equation}
\end{lemma}

\begin{proof}
With $\hat{F}_{Y}\left( y\right) $ given in (\ref{eq: kernel cdf}), let $%
\hat{f}_{Y}^{\left( i\right) }\left( y\right) =\hat{F}_{Y}^{\left(
i+1\right) }\left( y\right) $, for\ $i=1,2$, be the $i$th derivative of the
kernel marginal density estimator. Using standard methods for kernel
estimators (c.f. Robinson, 1983), we obtain under the assumptions of the
lemma that, as $n\rightarrow \infty ,h\rightarrow 0$, and $%
nh^{1+2i}\rightarrow \infty $, 
\begin{equation}
\sqrt{nh^{1+2i}}\left\{ \hat{f}_{Y}^{\left( i\right) }\left( y\right)
-f_{Y}^{\left( i\right) }\left( y\right) -\frac{1}{2}h^{2}\kappa
_{2}f_{Y}^{\left( i+2\right) }\left( y\right) \right\} \rightarrow
^{d}N\left( 0,V_{i}\left( y\right) \right)  \label{R1983}
\end{equation}%
where $V_{i}\left( y\right) =f_{Y}\left( y\right) \int_{\mathbb{R}}K^{\left(
i\right) }\left( z\right) ^{2}dz$. Assumptions 2.1 and 4.4 ensure that $%
f_{Y}\left( y\right) $ is sufficiently smooth so that $f_{Y}^{\left(
2\right) }\left( y\right) $\ and $f_{Y}^{\left( 3\right) }\left( y\right) $\
exist. Assumption 4.2(i) and 4.6 regulate the mixing property of $Y$ and the
kernel function, respectively, as required by Robinson (1983).

From (\ref{U1}) we have $\hat{U}^{\prime }\left( y\right) =\hat{f}_{Y}\left(
y\right) /f_{X}(\hat{U}\left( y\right) ;\hat{\theta})$. Now define $\hat{U}%
_{0}^{\prime }\left( y\right) =\hat{f}_{Y}\left( y\right) /f_{X}(U\left(
y\right) ;\theta _{0})$ and note that Assumption 4.4 and 4.5 together with
the delta-method\ imply $\hat{U}^{\prime }\left( y\right) -\hat{U}%
_{0}^{\prime }\left( y\right) =O_{P}\left( 1/\sqrt{n}\right) =o_{P}(1/\sqrt{%
nh})$. It then follows that%
\begin{eqnarray*}
&&\sqrt{nh}\left\{ \hat{U}^{\prime }\left( y\right) -U^{\prime }\left(
y\right) -\frac{1}{2}h^{2}\kappa _{2}f_{Y}^{\left( 2\right) }\left( y\right) 
\frac{1}{f_{X}\left( U\left( y\right) ;\theta _{0}\right) }\right\}  \\
&=&\sqrt{nh}\left\{ o_{P}\left( 1/\sqrt{nh}\right) +\hat{U}_{0}^{\prime
}\left( y\right) -U^{\prime }\left( y\right) -\frac{1}{2}h^{2}\kappa
_{2}f_{Y}^{\left( 2\right) }\left( y\right) \frac{1}{f_{X}\left( U\left(
y\right) ;\theta _{0}\right) }\right\}  \\
&=&\frac{1}{f_{X}\left( U\left( y\right) ;\theta _{0}\right) }\sqrt{nh}%
\left\{ \hat{f}_{Y}\left( y\right) -f_{Y}\left( y\right) -\frac{1}{2}%
h^{2}\kappa _{2}f_{Y}^{\left( 2\right) }\left( y\right) \right\}
+o_{P}\left( 1\right) .
\end{eqnarray*}%
Using (\ref{R1983}) and the same arguments as in Kristensen (2011, Proof of
Theorem 1), we arrive at (\ref{AD_U1}).

Next, observe that $U^{\prime \prime }\left( y\right) =\frac{f_{Y}^{\prime
}\left( y\right) }{f_{X}\left( U\left( y\right) ;\theta \right) }-\frac{%
f_{X}^{\prime }\left( U\left( y\right) ;\theta \right) f_{Y}\left( y\right)
^{2}}{f_{X}\left( U\left( y\right) ;\theta \right) ^{3}}$ where $%
f_{X}^{\prime }\left( x;\theta \right) $ and $f_{Y}^{\prime }\left( y\right) 
$ are the first derivatives of $f_{X}\left( x;\theta \right) $ and $%
f_{Y}\left( y\right) $, respectively. Similarly, it is easily checked that $%
\hat{U}^{\prime \prime }\left( y\right) =\frac{\hat{f}_{Y}^{\prime }\left(
y\right) }{f_{X}(\hat{U}\left( y\right) ;\hat{\theta})}-\frac{f_{X}^{\prime
}(\hat{U}\left( y\right) ;\hat{\theta})\hat{f}_{Y}\left( y\right) ^{2}}{%
f_{X}(\hat{U}\left( y\right) ;\hat{\theta})^{3}}$. Define $\hat{U}%
_{0}^{\prime \prime }\left( y\right) =\frac{\hat{f}_{Y}^{\prime }\left(
y\right) }{f_{X}\left( U\left( y\right) ;\theta _{0}\right) }-\frac{%
f_{X}^{\prime }\left( U\left( y\right) ;\theta _{0}\right) f_{Y}\left(
y\right) ^{2}}{f_{X}\left( U\left( y\right) ;\theta _{0}\right) ^{3}}$ and
apply arguments similar to before to obtain%
\begin{eqnarray*}
&&\sqrt{nh^{3}}\left\{ \hat{U}^{\prime \prime }\left( y\right) -U^{\prime
\prime }\left( y\right) -\frac{1}{2}h^{2}\kappa _{2}f_{Y}^{\left( 3\right)
}\left( y\right) \frac{1}{f_{X}\left( U\left( y\right) ;\theta _{0}\right) }%
\right\}  \\
&=&\frac{1}{f_{X}\left( U\left( y\right) ;\theta _{0}\right) }\sqrt{nh^{3}}%
\left\{ f_{Y}^{\prime }\left( y\right) -f_{Y}^{\prime }\left( y\right) -%
\frac{1}{2}h^{2}\kappa _{2}f_{Y}^{\left( 3\right) }\left( y\right) \right\}
+o_{p}\left( 1\right) 
\end{eqnarray*}%
which together with (\ref{R1983}) yield (\ref{AD_U2}).
\end{proof}

\pagebreak

\section{Tables and Figures}

\begin{center}
\textbf{Table 1: Bias and RMSE of }$\kappa $\textbf{\ in the OU-SKST Model}

\begin{tabular}{llcccc}
\hline\hline
&  & \multicolumn{4}{c}{Bias/$\kappa $} \\ \cline{3-6}
Sample Size &  & \multicolumn{2}{c}{2202} & \multicolumn{2}{c}{5505} \\ 
\cline{3-6}
True Parameter Value & $\rho _{1}$ & PPMLE & PMLE & PPMLE & PMLE \\ 
$\kappa =1.1376$ & $0.9944$ & \multicolumn{1}{l}{$0.6121$} & 
\multicolumn{1}{l}{$1.1379$} & \multicolumn{1}{l}{$0.2690$} & 
\multicolumn{1}{l}{$0.5054$} \\ 
$\kappa =5.6882$ & $0.9758$ & \multicolumn{1}{l}{$0.1230$} & 
\multicolumn{1}{l}{$0.1987$} & \multicolumn{1}{l}{$0.0652$} & 
\multicolumn{1}{l}{$0.0939$} \\ 
$\kappa =11.377$ & $0.9531$ & \multicolumn{1}{l}{$0.0656$} & 
\multicolumn{1}{l}{$0.0888$} & \multicolumn{1}{l}{$0.0400$} & 
\multicolumn{1}{l}{$0.0441$} \\ 
$\kappa =22.753$ & $0.9093$ & \multicolumn{1}{l}{$0.0385$} & 
\multicolumn{1}{l}{$0.0383$} & \multicolumn{1}{l}{$0.0270$} & 
\multicolumn{1}{l}{$0.0210$} \\ \hline
&  & \multicolumn{4}{c}{RMSE/$\kappa $} \\ \cline{3-6}
Sample Size &  & \multicolumn{2}{c}{2202} & \multicolumn{2}{c}{5505} \\ 
\cline{3-6}
True Parameter Value & $\rho _{1}$ & PPMLE & PMLE & PPMLE & PMLE \\ 
$\kappa =1.1376$ & $0.9944$ & \multicolumn{1}{l}{$0.8603$} & 
\multicolumn{1}{l}{$1.2932$} & \multicolumn{1}{l}{$0.4476$} & 
\multicolumn{1}{l}{$0.6224$} \\ 
$\kappa =5.6882$ & $0.9758$ & \multicolumn{1}{l}{$0.2420$} & 
\multicolumn{1}{l}{$0.2930$} & \multicolumn{1}{l}{$0.1454$} & 
\multicolumn{1}{l}{$0.1655$} \\ 
$\kappa =11.377$ & $0.9531$ & \multicolumn{1}{l}{$0.1574$} & 
\multicolumn{1}{l}{$0.1730$} & \multicolumn{1}{l}{$0.0974$} & 
\multicolumn{1}{l}{$0.1044$} \\ 
$\kappa =22.753$ & $0.9093$ & \multicolumn{1}{l}{$0.1059$} & 
\multicolumn{1}{l}{$0.1133$} & \multicolumn{1}{l}{$0.0668$} & 
\multicolumn{1}{l}{$0.0711$} \\ \hline
\end{tabular}

\pagebreak

\textbf{Table 2: Bias and RMSE of }$\kappa $\textbf{\ in the CIR-SKST Model}

\begin{tabular}{lllcccc}
\hline\hline
&  &  & \multicolumn{4}{c}{Bias/$\kappa $} \\ \cline{4-7}
Sample Size &  &  & \multicolumn{2}{c}{2202} & \multicolumn{2}{c}{5505} \\ 
\cline{4-7}
\multicolumn{2}{l}{True Parameter Values} & $\rho _{1}$ & PPMLE & PMLE & 
PPMLE & PMLE \\ 
\multicolumn{2}{l}{$\left( \kappa ,\alpha \right) =\left(
0.7653,1.1653\right) $} & $0.9921$ & \multicolumn{1}{l}{$0.9023$} & 
\multicolumn{1}{l}{$1.5269$} & \multicolumn{1}{l}{$0.4576$} & 
\multicolumn{1}{l}{$0.7717$} \\ 
\multicolumn{2}{l}{$\left( \kappa ,\alpha \right) =\left(
3.8267,1.1653\right) $} & $0.9675$ & \multicolumn{1}{l}{$0.2358$} & 
\multicolumn{1}{l}{$0.3347$} & \multicolumn{1}{l}{$0.1194$} & 
\multicolumn{1}{l}{$0.1754$} \\ 
\multicolumn{2}{l}{$\left( \kappa ,\alpha \right) =\left(
7.6533,1.1653\right) $} & $0.9399$ & \multicolumn{1}{l}{$0.1328$} & 
\multicolumn{1}{l}{$0.1816$} & \multicolumn{1}{l}{$0.0646$} & 
\multicolumn{1}{l}{$0.0853$} \\ 
\multicolumn{2}{l}{$\left( \kappa ,\alpha \right) =\left(
15.307,1.1653\right) $} & $0.8917$ & \multicolumn{1}{l}{$0.0768$} & 
\multicolumn{1}{l}{$0.0928$} & \multicolumn{1}{l}{$0.0349$} & 
\multicolumn{1}{l}{$0.0398$} \\ \hline
&  &  & \multicolumn{4}{c}{RMSE/$\kappa $} \\ \cline{4-7}
Sample Size &  &  & \multicolumn{2}{c}{2202} & \multicolumn{2}{c}{5505} \\ 
\cline{4-7}
\multicolumn{2}{l}{True Parameter Values} & $\rho _{1}$ & PPMLE & PMLE & 
PPMLE & PMLE \\ 
\multicolumn{2}{l}{$\left( \kappa ,\alpha \right) =\left(
0.7653,1.1653\right) $} & $0.9921$ & \multicolumn{1}{l}{$1.2424$} & 
\multicolumn{1}{l}{$1.7509$} & \multicolumn{1}{l}{$0.6692$} & 
\multicolumn{1}{l}{$0.9231$} \\ 
\multicolumn{2}{l}{$\left( \kappa ,\alpha \right) =\left(
3.8267,1.1653\right) $} & $0.9675$ & \multicolumn{1}{l}{$0.3881$} & 
\multicolumn{1}{l}{$0.4511$} & \multicolumn{1}{l}{$0.2363$} & 
\multicolumn{1}{l}{$0.2746$} \\ 
\multicolumn{2}{l}{$\left( \kappa ,\alpha \right) =\left(
7.6533,1.1653\right) $} & $0.9399$ & \multicolumn{1}{l}{$0.2431$} & 
\multicolumn{1}{l}{$0.2771$} & \multicolumn{1}{l}{$0.1498$} & 
\multicolumn{1}{l}{$0.1672$} \\ 
\multicolumn{2}{l}{$\left( \kappa ,\alpha \right) =\left(
15.307,1.1653\right) $} & $0.8917$ & \multicolumn{1}{l}{$0.1712$} & 
\multicolumn{1}{l}{$0.1847$} & \multicolumn{1}{l}{$0.1003$} & 
\multicolumn{1}{l}{$0.1068$} \\ \hline
\end{tabular}

\bigskip

\bigskip

\textbf{Table 3: Bias and RMSE of }$\alpha $\textbf{\ in the CIR-SKST Model}

\begin{tabular}{lllcccc}
\hline\hline
&  &  & \multicolumn{4}{c}{Bias/$\alpha $} \\ \cline{4-7}
Sample Size &  &  & \multicolumn{2}{c}{2202} & \multicolumn{2}{c}{5505} \\ 
\cline{4-7}
\multicolumn{2}{l}{True Parameter Values} & $\rho _{1}$ & PPMLE & PMLE & 
PPMLE & PMLE \\ 
\multicolumn{2}{l}{$\left( \kappa ,\alpha \right) =\left(
0.7653,1.1653\right) $} & $0.9921$ & \multicolumn{1}{l}{$0.9458$} & 
\multicolumn{1}{l}{$1.0299$} & \multicolumn{1}{l}{$0.6192$} & 
\multicolumn{1}{l}{$0.8720$} \\ 
\multicolumn{2}{l}{$\left( \kappa ,\alpha \right) =\left(
3.8267,1.1653\right) $} & $0.9675$ & \multicolumn{1}{l}{$0.4353$} & 
\multicolumn{1}{l}{$0.5171$} & \multicolumn{1}{l}{$0.1899$} & 
\multicolumn{1}{l}{$0.2554$} \\ 
\multicolumn{2}{l}{$\left( \kappa ,\alpha \right) =\left(
7.6533,1.1653\right) $} & $0.9399$ & \multicolumn{1}{l}{$0.2633$} & 
\multicolumn{1}{l}{$0.3152$} & \multicolumn{1}{l}{$0.1033$} & 
\multicolumn{1}{l}{$0.1279$} \\ 
\multicolumn{2}{l}{$\left( \kappa ,\alpha \right) =\left(
15.307,1.1653\right) $} & $0.8917$ & \multicolumn{1}{l}{$0.1302$} & 
\multicolumn{1}{l}{$0.1646$} & \multicolumn{1}{l}{$0.0663$} & 
\multicolumn{1}{l}{$0.0780$} \\ \hline
&  &  & \multicolumn{4}{c}{RMSE/$\alpha $} \\ \cline{4-7}
Sample Size &  &  & \multicolumn{2}{c}{2202} & \multicolumn{2}{c}{5505} \\ 
\cline{4-7}
\multicolumn{2}{l}{True Parameter Values} & $\rho _{1}$ & PPMLE & PMLE & 
PPMLE & PMLE \\ 
\multicolumn{2}{l}{$\left( \kappa ,\alpha \right) =\left(
0.7653,1.1653\right) $} & $0.9921$ & \multicolumn{1}{l}{$1.5614$} & 
\multicolumn{1}{l}{$1.5784$} & \multicolumn{1}{l}{$1.1222$} & 
\multicolumn{1}{l}{$1.4309$} \\ 
\multicolumn{2}{l}{$\left( \kappa ,\alpha \right) =\left(
3.8267,1.1653\right) $} & $0.9675$ & \multicolumn{1}{l}{$0.8443$} & 
\multicolumn{1}{l}{$0.9197$} & \multicolumn{1}{l}{$0.3867$} & 
\multicolumn{1}{l}{$0.4462$} \\ 
\multicolumn{2}{l}{$\left( \kappa ,\alpha \right) =\left(
7.6533,1.1653\right) $} & $0.9399$ & \multicolumn{1}{l}{$0.5473$} & 
\multicolumn{1}{l}{$0.5695$} & \multicolumn{1}{l}{$0.2298$} & 
\multicolumn{1}{l}{$0.2558$} \\ 
\multicolumn{2}{l}{$\left( \kappa ,\alpha \right) =\left(
15.307,1.1653\right) $} & $0.8917$ & \multicolumn{1}{l}{$0.2802$} & 
\multicolumn{1}{l}{$0.3139$} & \multicolumn{1}{l}{$0.1453$} & 
\multicolumn{1}{l}{$0.1684$} \\ \hline
\end{tabular}

\pagebreak

\textbf{Table 4: Descriptive Statistics of Daily VIX}

\begin{tabular}{lc}
\hline\hline
Sample Period & January 2, 1990 - July 19, 2019 \\ 
Sample Size & 7445 \\ 
Mean & 19.21 \\ 
Median & 17.31 \\ 
Std Dev. & 7.76 \\ 
Skewness & 2.12 \\ 
Kurtosis & 10.85 \\ 
Jarque-Bera Statistic & 24669.26 \\ \hline
\end{tabular}

\bigskip

\bigskip

\textbf{Table 5: Model Estimation and Pseudo-LR Test Results}

\begin{tabular}{lcccc}
\hline\hline
& \multicolumn{2}{c}{Transformed OU} & \multicolumn{2}{c}{Transformed CIR}
\\ \hline
& DO & NPTOU & EW & NPTCIR \\ 
$\hat{\kappa}$ & 4.4888 & 3.8191 & 4.0741 & 3.7541 \\ 
& (0.5795) & (0.4525) & (0.5597) & (0.4257) \\ 
$\hat{\alpha}$ & 2.8890 &  & 0.0524 & 14.6916 \\ 
& (0.0423) &  & (0.0032) & (8.8484) \\ 
$\hat{\sigma}^{2}$ & 1.0818 &  & 0.0695 &  \\ 
& (0.0179) &  & (0.0097) &  \\ 
$\hat{\varrho}$ &  &  & 0.1916 &  \\ 
&  &  & (0.4827) &  \\ 
$\hat{\delta}$ &  &  & 0.0072 &  \\ 
&  &  & (0.0029) &  \\ \hline
$LL\left( 10^{4}\right) $ & -1.1724 & -1.1579 & -1.1585 & -1.1565 \\ 
$LR$ &  & 290.7263 &  & 40.8606 \\ 
$CV_{0.05}$ &  & -52.1521 &  & -23.6766 \\ 
$CV_{0.01}$ &  & -30.5511 &  & -10.9027 \\ 
$p$-value &  & 0.0000 &  & 0.0000 \\ \hline
\end{tabular}

\begin{figure}
	\centering
	\includegraphics[width=0.7\linewidth]{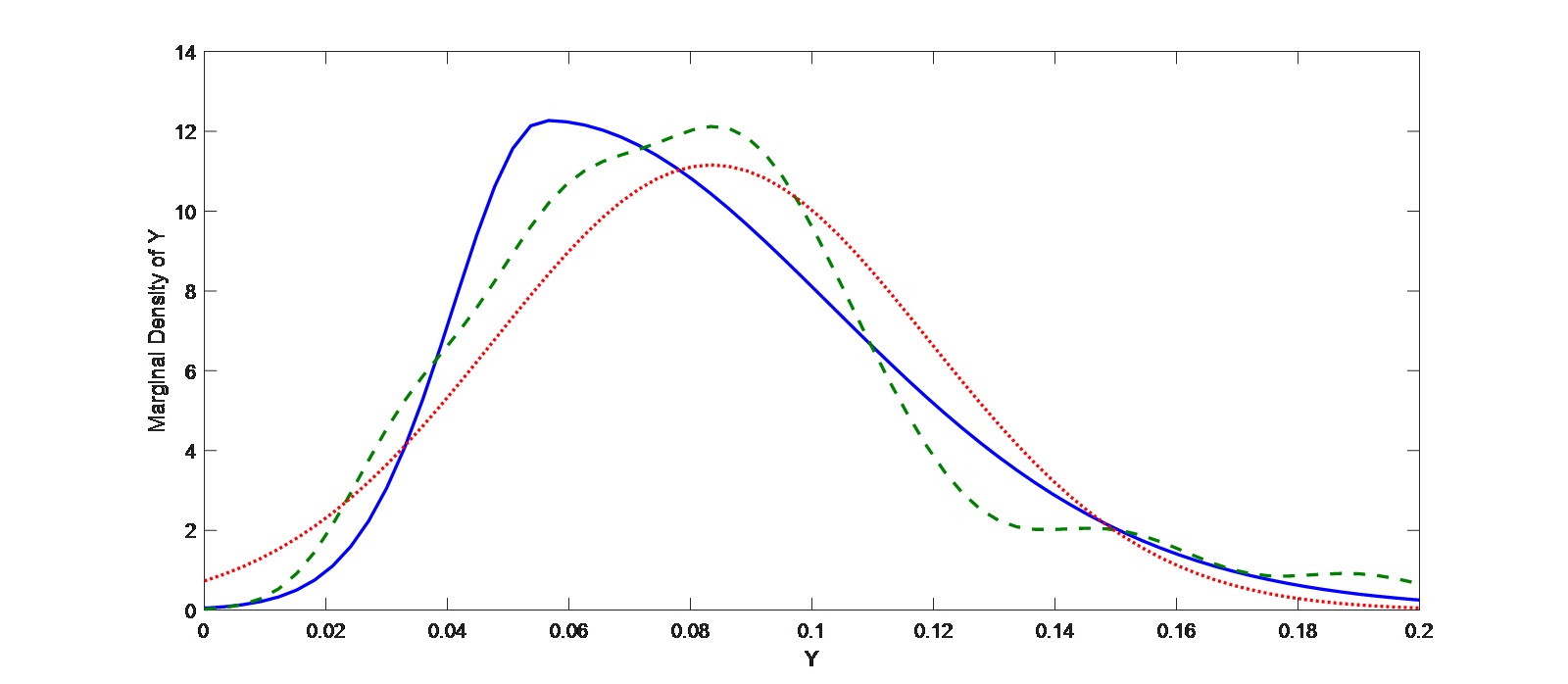}
	\caption[]{Marginal Densities of the Eurodollar Rates. \newline {\small Solid = SKST Density, Dashed = Kernel Density, Dotted = Normal Density}}
\label{fig:marginaldensities}
\end{figure}

\begin{figure}
	\centering
	\includegraphics[width=0.7\linewidth]{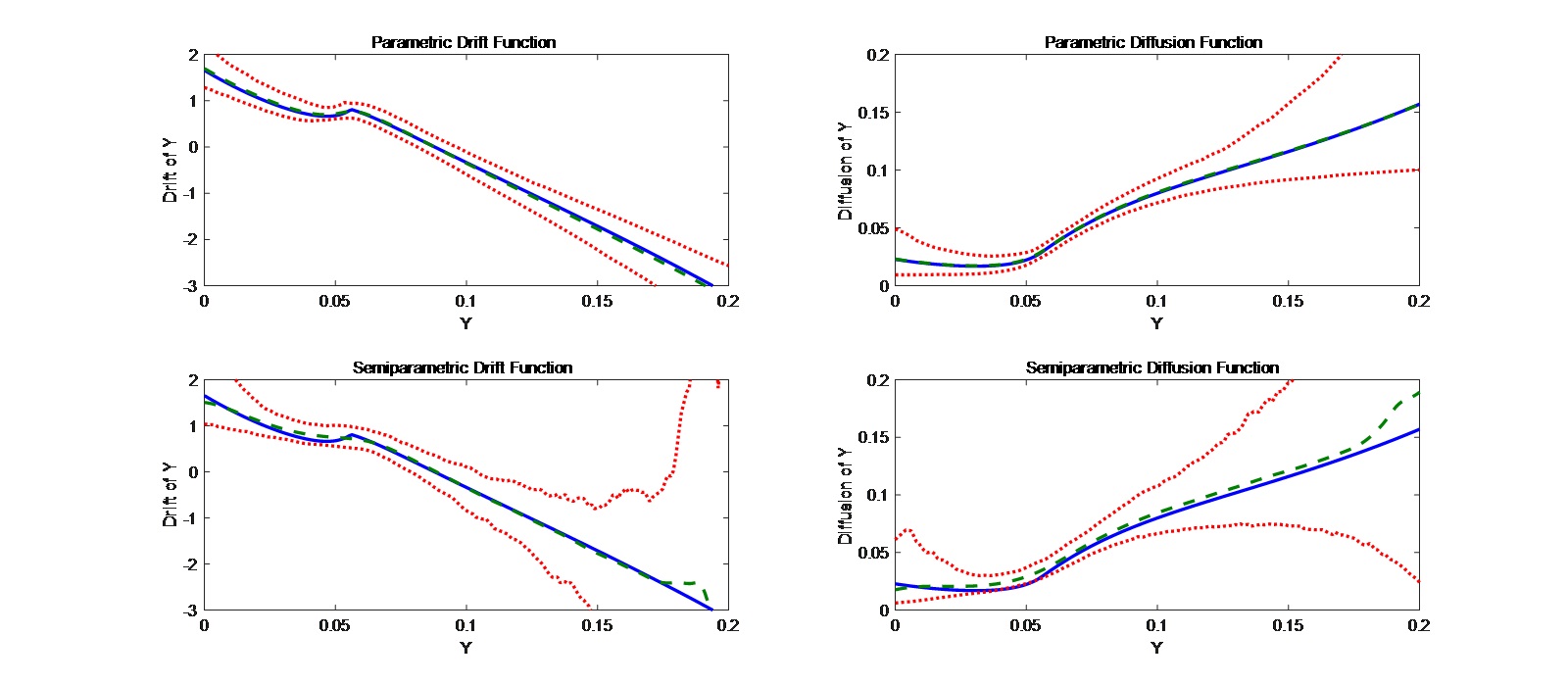}
	\caption[]{Estimated Drift and Diffusion for the OU-SKST Model ($T=2202$) . \newline {\small Solid = True Function, Dashed = Mean of Estimates, Dotted = 95\% Confidence Bands}}
	\label{fig:ouskst2202}
\end{figure}

\begin{figure}
	\centering
	\includegraphics[width=0.7\linewidth]{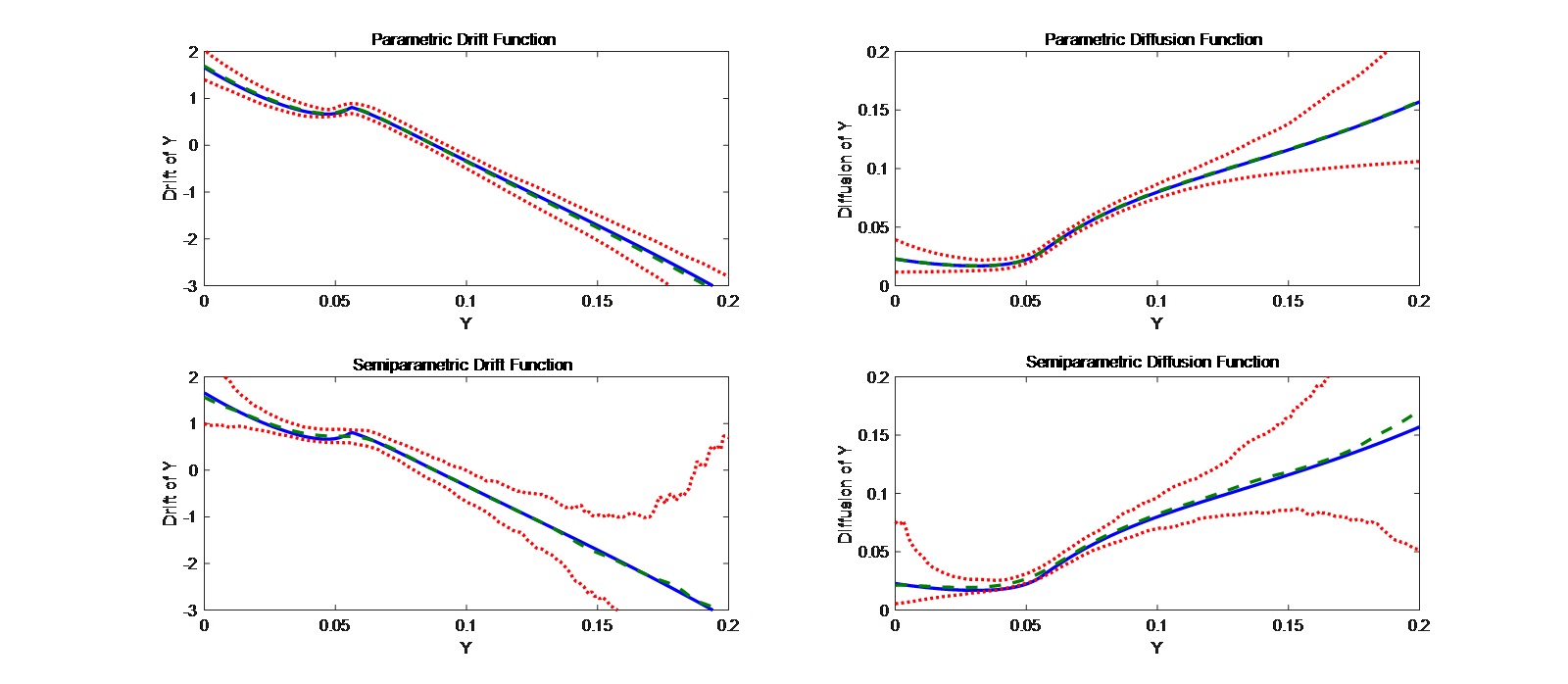}
	\caption[]{Estimated Drift and Diffusion for the OU-SKST Model ( $T=5505$). \newline {\small Solid = True Function, Dashed = Mean of Estimates, Dotted = 95\% Confidence Bands}}
	\label{fig:ouskst5505}
\end{figure}

\begin{figure}
	\centering
	\includegraphics[width=0.7\linewidth]{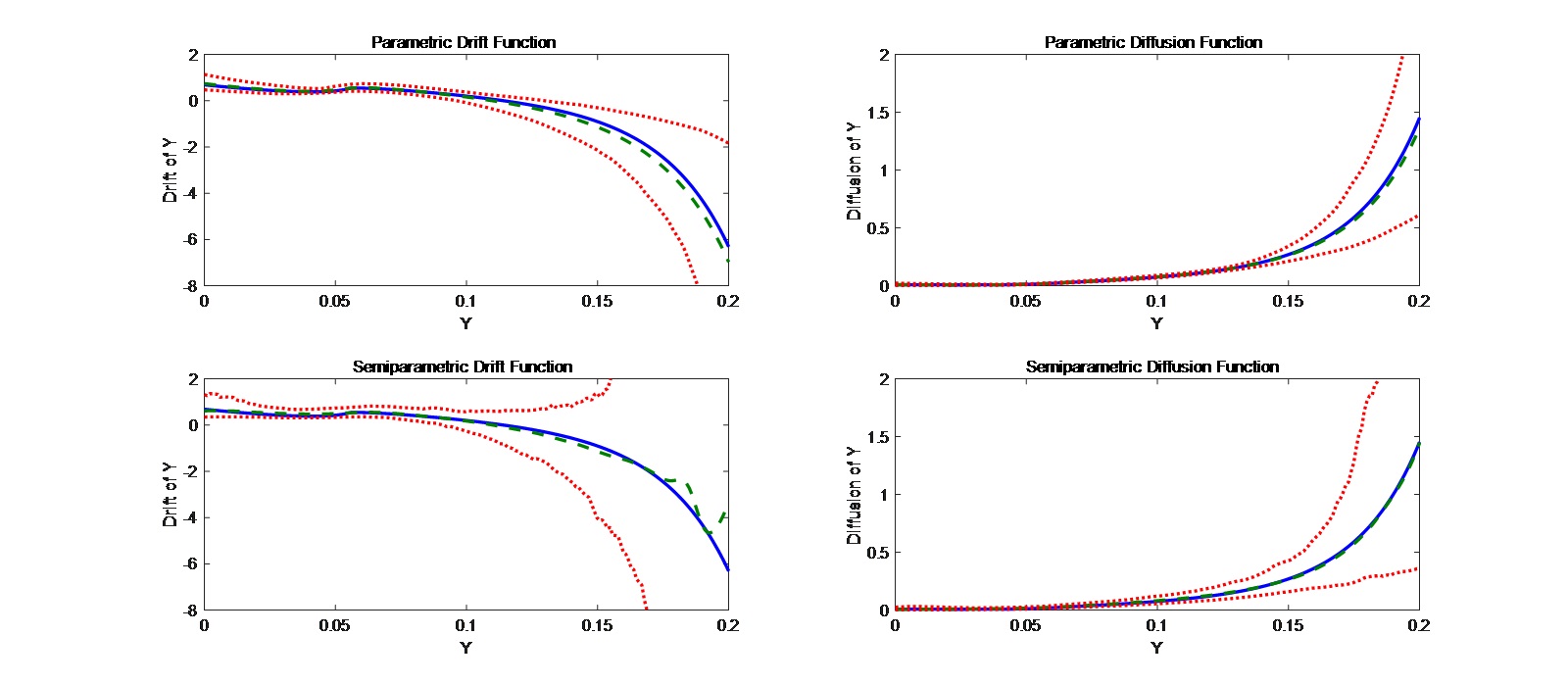}
	\caption[]{Estimated Drift and Diffusion for the CIR-SKST Model ($ T=2202$). \newline {\small Solid = True Function, Dashed = Mean of Estimates, Dotted = 95\% Confidence Bands}}
	\label{fig:cirskst2202}
\end{figure}

\begin{figure}
	\centering
	\includegraphics[width=0.7\linewidth]{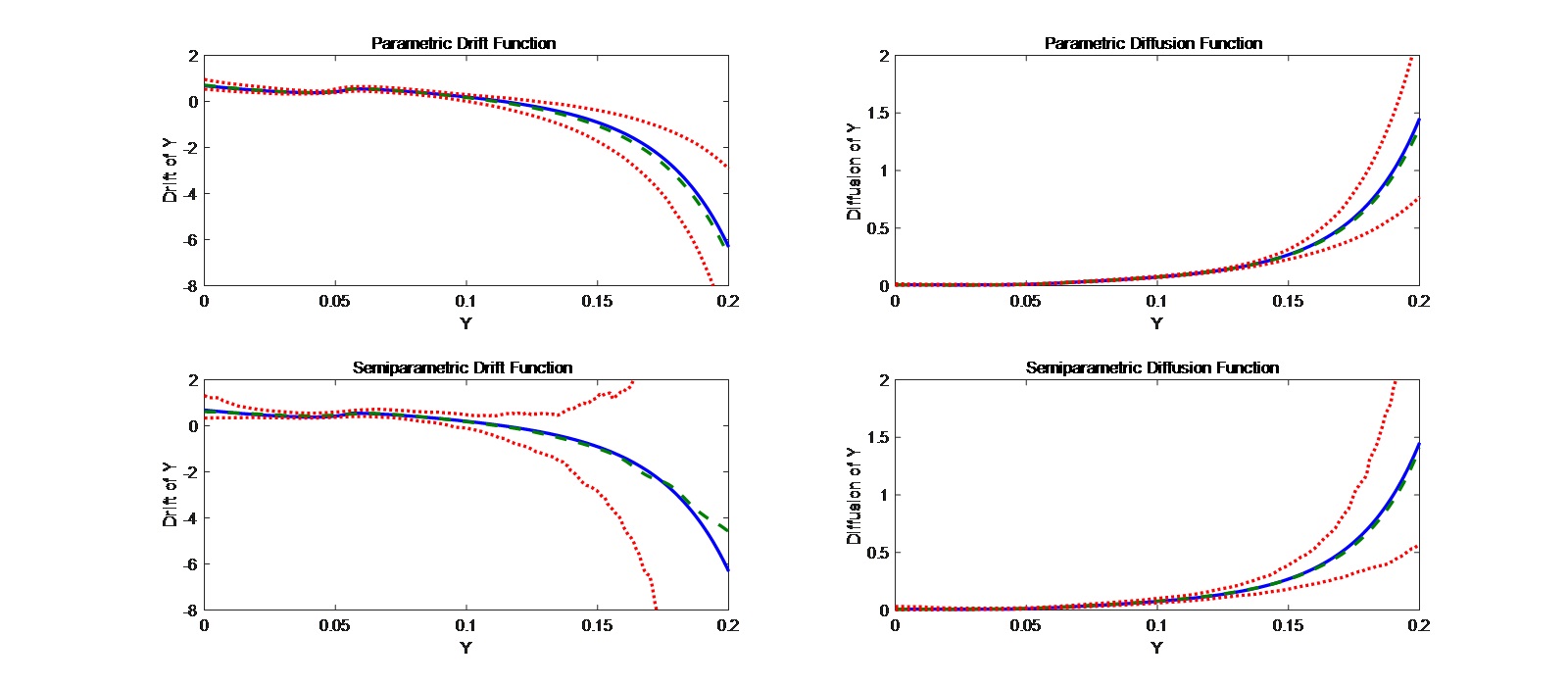}
	\caption[]{Estimated Drift and Diffusion for the CIR-SKST Model ($ T=5505$). \newline {\small Solid = True Function, Dashed = Mean of Estimates, Dotted = 95\% Confidence Bands}}
	\label{fig:cirskst5505}
\end{figure}

\begin{figure}
	\centering
	\includegraphics[width=0.7\linewidth]{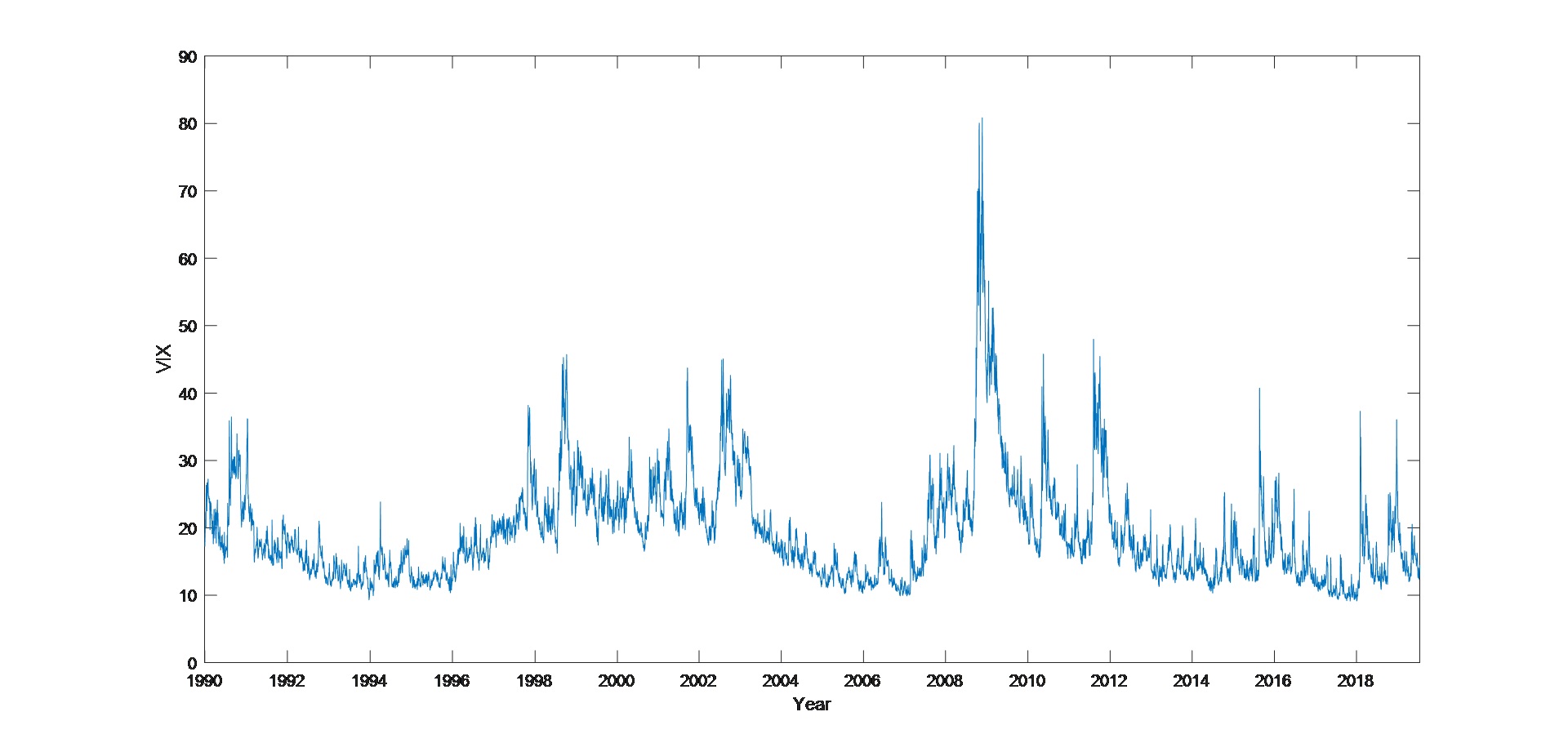}
	\caption[]{Time Series of Daily VIX}
	\label{fig:q9i8lf00}
\end{figure}

\begin{figure}
	\centering
	\includegraphics[width=0.7\linewidth]{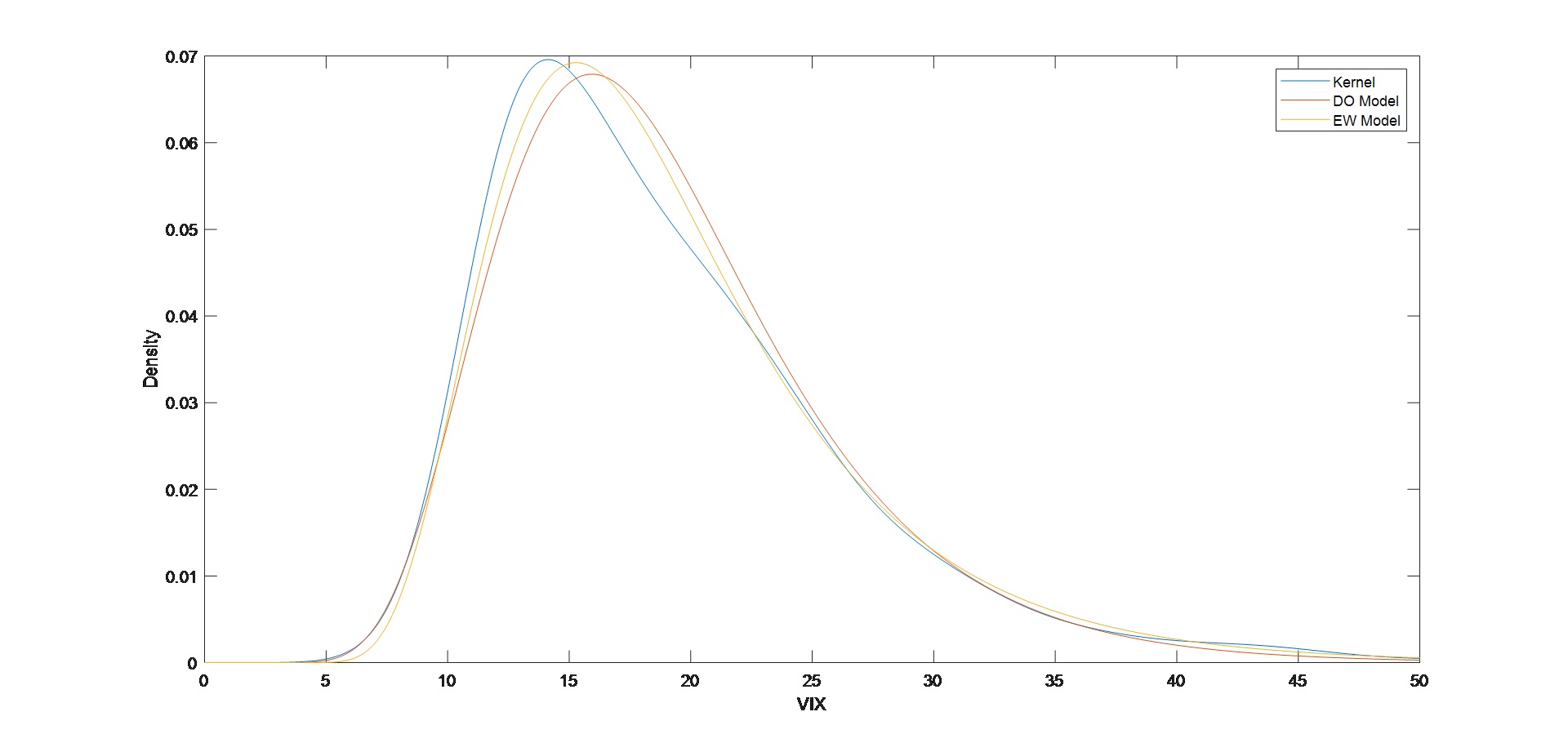}
	\caption[]{Estimated Marginal Densities of Daily VIX}
	\label{fig:q9i8lf01}
\end{figure}

\end{center}

\end{document}